\newif\ifijuq
\ijuqfalse   

\newif\iflongappendix
\longappendixtrue    

\ifijuq
  \documentclass{ijuq}
\else
  \documentclass[11pt]{article}
  \usepackage[margin=1in]{geometry}
  \usepackage{amsmath,amssymb,amsthm}
  \usepackage{mathtools}
\usepackage{bm}
  \usepackage[numbers,sort&compress]{natbib}  
  \usepackage[colorlinks=true,
              linkcolor=blue,
              citecolor=blue,
              urlcolor=blue]{hyperref}
\fi

\usepackage{graphicx}
\usepackage{booktabs}        
\usepackage{microtype}       
\usepackage{xcolor}
\usepackage[boxruled]{algorithm2e}     
  \SetAlgoLined
  \SetKwComment{Comment}{$\triangleright$\ }{}
  \SetKwComment{tcp}{// }{}
\usepackage{cleveref}        
\usepackage{todonotes}                             

\ifijuq
\else
  \newtheorem{theorem}{Theorem}[section]
  \newtheorem{proposition}[theorem]{Proposition}

  \theoremstyle{remark}
\fi

\newcommand{\vv}[1]{\boldsymbol{#1}}   
\newcommand{\mm}[1]{\bm{\mathrm{#1}}}  

\newcommand{\params}{\vv{\theta}}

\newcommand{\design}{\vv{w}}
\newcommand{\designk}[1]{w_{#1}}
\newcommand{\weightmtx}{\mm{W}}
\newcommand{\simplex}{\Delta^{\nobs}}

\newcommand{\qoidomain}{\mathcal{X}_\mathrm{qoi}}

\newcommand{\data}{\vv{y}}
\newcommand{\datam}[1]{\data^{(#1)}}

\newcommand{\noisevec}{\vv{\varepsilon}}
\newcommand{\noisem}[1]{\noisevec^{(#1)}}

\newcommand{\obsmap}{f}
\newcommand{\obsmapeval}{\obsmap(\params,\,\design)}
\newcommand{\obsspace}{\mathcal{Y}}

\newcommand{\fwdmodel}{\obsmap}

\newcommand{\qoimap}{q}
\newcommand{\qoi}{q(\params)}
\newcommand{\qoivec}{\vv{q}}
\newcommand{\qoicomp}[1]{q_{#1}(\params)}

\newcommand{\qoispace}{\mathcal{Q}}
\newcommand{\paramspace}{\Theta}
\newcommand{\desspace}{\mathcal{W}}

\newcommand{\prior}{\pi}

\newcommand{\likeli}{p(\data \mid \params,\, \design)}

\newcommand{\post}{p(\params \mid \data,\, \design)}

\newcommand{\Gaussian}[2]{\mathcal{N}\!\left(#1,\,#2\right)}
\newcommand{\LogNormal}[2]{\mathcal{LN}\!\left(#1,\,#2\right)}

\newcommand{\obsmtx}{\mm{A}}

\newcommand{\noisecov}{\mm{\Gamma}}
\newcommand{\effnoisecov}{\mm{\Gamma}_\mathrm{eff}(\design)}

\newcommand{\priorcov}{\mm{\Sigma}}
\newcommand{\priormean}{\vv{\mu}}

\newcommand{\postcov}{\mm{\Sigma}_{\!*}(\design)}
\newcommand{\postcovb}{\mm{\Sigma}_{\!*}}
\newcommand{\postmean}{\vv{\mu}_*(\data,\,\design)}
\newcommand{\postmeanb}{\vv{\mu}_*}

\newcommand{\pmnu}{\vv{\nu}}           
\newcommand{\pmcov}{\mm{C}}            

\newcommand{\predvec}{\vv{\psi}}
\newcommand{\predmtx}{\mat{\Psi}}

\newcommand{\fisherb}{\mm{F}}

\newcommand{\devmeas}{\mathcal{D}^{\mathrm{pf}}}                              
\newcommand{\devmeasp}{\mathcal{D}^{\mathrm{pf}}_{p(\params \mid \data, \design)}}  

\newcommand{\devk}[1]{d_{#1}(\data,\design)}   
\newcommand{\devkb}[1]{d_{#1}}                   

\newcommand{\qoirisk}{\mathcal{R}^{\mathrm{QoI}}}                             
\newcommand{\qoiriskm}{\mathcal{R}^{\mathrm{QoI}}_{\mu}}                      
\newcommand{\levtwo}{r(\data,\design)}   
\newcommand{\levtwob}{r}                  

\newcommand{\datarisk}{\mathcal{R}^{\mathrm{data}}}                           
\newcommand{\datariskd}{\mathcal{R}^{\mathrm{data}}_{p(\data \mid \design)}}  

\newcommand{\utility}{U(\design)}
\newcommand{\utilityb}{U}

\newcommand{\kluty}{U_{\mathrm{KL}}(\design)}    

\newcommand{\outwt}[1]{\omega_{#1}^{\mathrm{out}}}
\newcommand{\inwt}[1]{\omega_{#1}^{\mathrm{in}}}

\newcommand{\outersamp}[1]{\params^{(#1)}}
\newcommand{\innersamp}[1]{\params^{(#1)}}

\newcommand{\impwt}[2]{\tilde{w}_{#2}^{(#1)}(\design)}

\newcommand{\residsym}{e}
\newcommand{\residk}[3]{\residsym_{#1#2#3}}

\newcommand{\E}[1]{\mathbb{E}\!\left[#1\right]}
\newcommand{\Eunder}[2]{\mathbb{E}_{#1}\!\left[#2\right]}
\newcommand{\Stdunder}[2]{\mathrm{Std}_{#1}\!\left[#2\right]}
\newcommand{\Varunder}[2]{\mathrm{Var}_{#1}\!\left[#2\right]}
\newcommand{\Var}[1]{\mathbb{V}\!\left[#1\right]}
\newcommand{\Std}[1]{\mathrm{Std}\!\left[#1\right]}
\newcommand{\Cov}[2]{\mathrm{Cov}\!\left[#1,\,#2\right]}
\newcommand{\AVaR}[2]{\mathrm{AVaR}_{#1}\!\left[#2\right]}
\newcommand{\VaR}[2]{\mathrm{VaR}_{#1}\!\left[#2\right]}
\newcommand{\KL}[2]{\mathrm{KL}\!\left(#1 \,\|\, #2\right)}
\newcommand{\EIG}{\mathrm{EIG}}
\newcommand{\normcdf}[1]{\Phi\!\left(#1\right)}

\newcommand{\normppf}[1]{\Phi^{-1}\!\left(#1\right)}

\newcommand{\nobs}{K}        
\newcommand{\nparams}{D}     
\newcommand{\nqoi}{Q}        
\newcommand{\nout}{M}        
\newcommand{\nin}{N}         

\newcommand{\lnscale}{K}
\newcommand{\lnscalej}[1]{K_{#1}}

\newcommand{\lntau}{\tau}
\newcommand{\lntauj}[1]{\tau_{#1}}

\newcommand{\lnsigmasq}{\sigma^2}
\newcommand{\lnsigmasqj}[1]{\sigma_{#1}^2}

\newcommand{\lnnu}{\nu}
\newcommand{\lnnuj}[1]{\nu_{#1}}

\newcommand{\lnsigmatau}{\sigma_\tau^2}
\newcommand{\lnsigmatauj}[1]{\sigma_{\tau,#1}^2}

\newcommand{\secref}[1]{Section~\ref{#1}}
\newcommand{\figref}[1]{Figure~\ref{#1}}
\newcommand{\tabref}[1]{Table~\ref{#1}}
\newcommand{\eqnref}[1]{(\ref{#1})}
\newcommand{\appref}[1]{Appendix~\ref{#1}}

\newcommand{\T}{^{\top}}

\newcommand{\inv}[1]{{#1}^{-1}}

\newcommand{\diag}{\mathrm{diag}}

\newcommand{\R}{\mathbb{R}}

\newcommand{\tr}{\mathrm{tr}}

\DeclareMathOperator*{\argmin}{arg\,min}

\newcommand{\Algref}[1]{Algorithm~\ref{#1}}

\usepackage{bm}   

\newcommand{\mat}[1]{\bm{\mathrm{#1}}}

\newcommand{\descond}{\xi}
\newcommand{\vdescond}{\vv{\xi}}
\newcommand{\descondspace}{\Xi}

\begin{document}

\ifijuq
  \title{Risk-Aware Goal-Oriented Bayesian Optimal Experimental Design}

  \author{
    John D.\ Jakeman\affil{1},
    Rebekah White\affil{1},
    Bart van Bloemen Waanders\affil{2},
    Drew P.\ Kouri\affil{1},
    Alen Alexanderian\affil{3}
  }

  \affiliation{1}{
    Optimization and Uncertainty Quantification Department,
    Sandia National Laboratories,
    Albuquerque, NM, USA
  }

  \affiliation{2}{
    Scientific Machine Learning Department,
    Sandia National Laboratories,
    Albuquerque, NM, USA
  }

  \affiliation{3}{
    Department of Mathematics,
    North Carolina State University,
    Raleigh, NC, USA
  }

  \maketitle

  \begin{abstract}
    Traditional Bayesian optimal experimental design (OED) selects measurements that best inform a model's parameters. However, such measurements can be suboptimal for downstream predictions.
Goal-oriented OED targets the prediction directly. However, the existing goal-oriented criteria value all reductions in predictive uncertainty equally, with no way to prioritize rare, high-consequence outcomes.
In this article, we develop a risk-aware framework that composes risk at three levels, each generalizing an ingredient of classical $I$- and $G$-optimal design: a deviation measure of the posterior predictive uncertainty (generalizing the predictive variance), a risk measure across the prediction domain (interpolating $I$-optimal averaging and $G$-optimal worst-case selection), and a risk measure over datasets (generalizing the expectation).
We generate each level from a regret function in the risk quadrangle, so that one triple specifies a practitioner's risk preference.
We relax the design to continuous weights on the unit simplex and construct a nested-quadrature estimator that is differentiable in the design variable. This enables solving the optimal design problem with gradient-based methods, avoiding a combinatorial search over candidate designs.
For a linear-Gaussian lognormal model and a nonlinear extension, we derive closed-form objectives. These give exact references against which we verify that the estimator converges.
We demonstrate this framework for finding optimal sensor placements in an inverse problem governed by an advection--diffusion equation.
We find that the risk-aware designs substantially outperform the expected-information-gain baseline, which is statistically indistinguishable from a random allocation.

  \end{abstract}


\else
  \title{
    Risk-Aware Goal-Oriented Bayesian\\
    Optimal Experimental Design
  }

  \author{
    John D.\ Jakeman\thanks{Corresponding author.
      AI Credibility Department,
      Sandia National Laboratories, Albuquerque, NM, USA.
      \texttt{jdjakem@sandia.gov}}
    \and
    Rebekah White\thanks{Optimization and Uncertainty Quantification Department,
      Sandia National Laboratories, Albuquerque, NM, USA.}
    \and
    Bart van Bloemen Waanders\thanks{Scientific Machine Learning Department,
      Sandia National Laboratories, Albuquerque, NM, USA.}
    \and
    Drew P.\ Kouri\footnotemark[2]
    \and
    Alen Alexanderian\thanks{Department of Mathematics,
      North Carolina State University, Raleigh, NC, USA.}
  }

  \date{\today}

  \maketitle

  \begin{abstract}
    Traditional Bayesian optimal experimental design (OED) selects measurements that best inform a model's parameters. However, such measurements can be suboptimal for downstream predictions.
Goal-oriented OED targets the prediction directly. However, the existing goal-oriented criteria value all reductions in predictive uncertainty equally, with no way to prioritize rare, high-consequence outcomes.
In this article, we develop a risk-aware framework that composes risk at three levels, each generalizing an ingredient of classical $I$- and $G$-optimal design: a deviation measure of the posterior predictive uncertainty (generalizing the predictive variance), a risk measure across the prediction domain (interpolating $I$-optimal averaging and $G$-optimal worst-case selection), and a risk measure over datasets (generalizing the expectation).
We generate each level from a regret function in the risk quadrangle, so that one triple specifies a practitioner's risk preference.
We relax the design to continuous weights on the unit simplex and construct a nested-quadrature estimator that is differentiable in the design variable. This enables solving the optimal design problem with gradient-based methods, avoiding a combinatorial search over candidate designs.
For a linear-Gaussian lognormal model and a nonlinear extension, we derive closed-form objectives. These give exact references against which we verify that the estimator converges.
We demonstrate this framework for finding optimal sensor placements in an inverse problem governed by an advection--diffusion equation.
We find that the risk-aware designs substantially outperform the expected-information-gain baseline, which is statistically indistinguishable from a random allocation.

  \end{abstract}

\fi


\section{Introduction}
Mathematical models are used to forecast the behavior of physical and engineered systems and to inform high-consequence decisions, such as whether a contaminant will exceed a regulatory threshold or a component will fail under service loading.
Uncertainty in model inputs and parameters, such as boundary conditions, material properties, and source terms, propagates to uncertainty in the model outputs that drive these decisions.
We call these outputs quantities of interest (QoIs).
Reducing uncertainty in these QoIs typically requires collecting experimental data to update the prior distribution over uncertain inputs to a posterior distribution. 
Pushing this posterior through the model (the posterior push-forward) yields the induced posterior distribution over the QoIs.
Measurements are limited by cost, time, and safety. Further, not all measurements are equally informative. Therefore, the choice of what to measure matters.
In the Bayesian setting, optimal experimental design (OED) formalizes this choice by selecting, \textit{a priori}, the measurements that maximize the expected information gained under a prescribed budget.

A common approach to Bayesian OED is to optimize a scalar function of the posterior distribution.
Classical criteria, such as $A$- and $D$-optimality, or the expected Kullback--Leibler (KL) divergence from the posterior to the prior, target the parameters, choosing measurements that make the parameter posterior as concentrated as possible~\cite{Chaloner_V_IMS_1995,Ryan_DMP_ISR_2016,Alexanderian_IP_2021,Rainforth_FIB_SS_2024}.
While these criteria are effective for parameter estimation, many applications focus on downstream tasks, like predicting the QoIs.
In this setting, parameter-focused OED (pOED) can expend effort on parameter directions only weakly coupled to the QoIs, yielding only marginal gains in predictive performance~\cite{Butler_JW_JCP_2020} and motivating alternative approaches~\cite{donatelli2025basicresearch}.

Goal-oriented optimal experimental design (gOED) addresses this mismatch by selecting experiments that reduce predictive uncertainty directly.
Rather than concentrating the parameter posterior, gOED targets the posterior push-forward—the distribution of the QoI obtained by mapping the parameter posterior through the parameter-to-QoI map.
Figure~\ref{fig:oed-concept} illustrates the resulting difference for sensor placement in a channel: pOED places sensors where they best constrain the inflow parameters, but those parameters are only weakly coupled to the downstream QoI, so its posterior push-forward barely tightens; gOED instead places sensors along the flow path to the receptor, markedly reducing predictive uncertainty.

\begin{figure}[htb]
    \centering
    \includegraphics[width=1\linewidth]{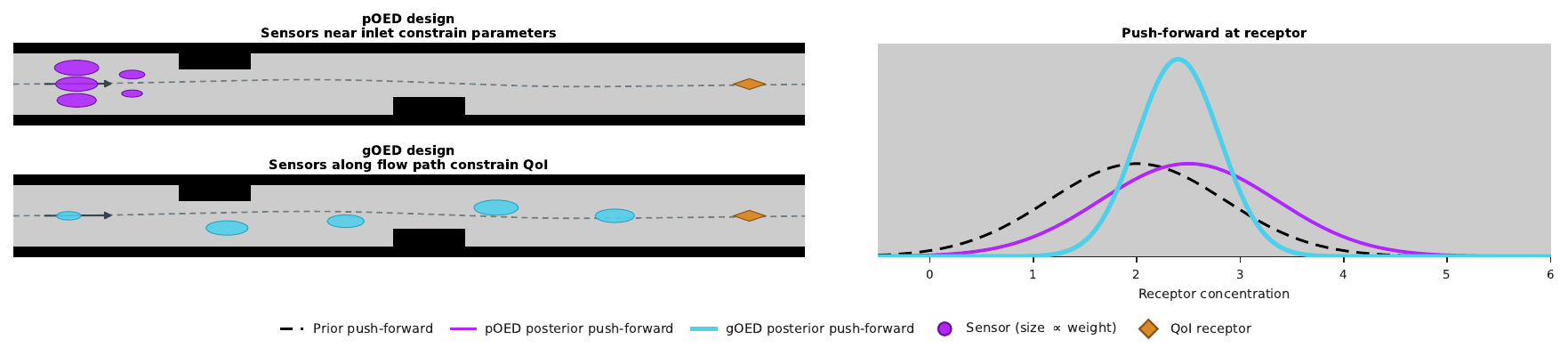}
    \caption{Conceptual comparison of pOED and gOED for sensor placement in a channel with flow obstruction.
    \emph{Left}: candidate sensor locations (circles, sized by design weight) and receptor location (diamond) defining the QoI; the pOED design (top-left) and gOED design (bottom-left) allocate weight to different sensors.
    \emph{Right}: prior (dashed) and posterior (solid) push-forward distributions of the QoI at the receptor under each design.}
    \label{fig:oed-concept}
\end{figure}

Existing gOED methods fall into two classes.
Laplace-based methods optimize a functional of the posterior push-forward covariance (typically its trace or determinant) averaged over plausible data under a Gaussian approximation of the posterior~\cite{Wu_CG_SISC_2023,Attia_A_S_IP_2018}.
In contrast, information-gain methods maximize the expected KL divergence between the prior and posterior push-forward, accommodating nonlinear and non-Gaussian models~\cite{Chakraborty_HC_Arxiv_2024,Zhong_SCH_Arxiv_2025}.
However, both cannot be tailored to a practitioner's risk preferences.
The covariance penalizes deviations on both sides of the mean equally, failing to prioritize the upper tail even when only large QoI values are of concern.
The KL divergence rewards reducing overall uncertainty in the posterior distribution, with no way to favor one region over another.
Moreover, both criteria are risk-neutral with respect to data uncertainty. Specifically, they evaluate performance by averaging over all plausible datasets, so rare but undesirable datasets are weighted the same as typical datasets.

In many applications, this insensitivity to risk is inadequate: the appropriate design criterion must respond to predictive uncertainty asymmetrically, emphasizing the outcomes that carry the highest decision cost.
For example, a reliability engineer cares about the probability that peak stress exceeds a failure limit, not about tightening the distribution around the mean of the Lam\'e  parameters.
Risk measures from decision theory, such as average value-at-risk (AVaR, also known as conditional value-at-risk) and entropic risk, provide an axiomatic way to encode such preferences and weight tail behavior accordingly.
Risk measures have proven valuable in surrogate modeling~\cite{Jakeman_KH_RESS_2021} and optimization under uncertainty~\cite{Kouri_S_SIAMOPT_2016}, but risk-aware criteria have seen limited development in OED and have been explored only in frequentist formulations~\cite{Kouri_JH_SIAMUQ_2021,peng2024efficient} that rely on large-sample approximations, which may be ill-suited for small observational budgets.
This gap motivates the framework developed in this paper, which integrates predictive-focused OED with a principled treatment of risk.

\emph{The main contribution of this work is a risk-aware optimality criterion for goal-oriented Bayesian OED that composes risk at three levels, each generalizing an ingredient of classical I- and G-optimal design.}
Three levels are required because goal-oriented Bayesian OED must:
(i) quantify dispersion in the posterior push-forward conditional on data, which we do using a deviation measure (e.g., standard deviation), generalizing the predictive variance underlying classical $I$- and $G$-optimality;
(ii) aggregate performance across components of a multivariate QoI, which we do using a risk measure that unifies $I$-optimal averaging and $G$-optimal worst-case aggregation; and
(iii) account for uncertainty in the dataset itself at design time, which we do using a risk measure over plausible data realizations, generalizing expectation-based methods and enabling designs that penalize anomalous observations for robustness.
At each level, we use the risk quadrangle~\cite{Rockafellar_U_SORMS_2013} to map a single regret function, which quantifies the displeasure associated with each possible outcome, to a matched deviation measure and risk measure.

To solve the sensor-design optimization problem induced by our risk-aware
goal-oriented OED criteria, we build on the nested ``double-loop'' Monte
Carlo estimators commonly used in Bayesian OED and reviewed in~\cite{Ryan_DMP_ISR_2016}.
These estimators use an outer loop to draw plausible datasets under a candidate design and an inner loop to sample the corresponding parameter posterior in order to evaluate a parameter-focused utility. 
We generalize this nested-estimation template to goal-oriented, risk-aware design by evaluating utilities of the posterior \emph{push-forward} to QoIs and by using risk measures tailored to a practitioner’s risk preferences.
Unlike existing goal-oriented OED approaches, this estimator avoids both the linearization of covariance-based criteria and the kernel density estimation of information-gain criteria, the latter degrading rapidly beyond a few QoIs and breaking down for continuous prediction fields~\cite{Terrel_S_AS_1992}.

Since practical designs involve discrete sensor selection and integer measurement allocations, the exact problem is combinatorial. We relax these constraints by allowing continuous weights on the unit simplex and reuse a fixed set of Monte Carlo samples, so the objective becomes a differentiable function of the design weights.
We further derive closed-form expressions for several risk-aware design objectives in a linear-Gaussian model with lognormal push-forward, which serve as exact references for verifying the convergence of the nested estimators.
Finally, we demonstrate the framework on a PDE-constrained advection--diffusion problem representative of surface-water quality monitoring, showing that different risk preferences yield distinct sensor placements and that parameter-focused designs can be suboptimal for goal-oriented objectives.

The rest of the paper is organized as follows.
\Cref{sec:problem_formulation} formulates the Bayesian inverse problem, its goal-oriented variant, and the resulting OED problem, with an illustrative linear-Gaussian example used throughout the paper.
\Cref{sec:utility} introduces risk measures and uses them to form risk-aware OED objectives.
\Cref{sec:computation} presents the numerical algorithm for computing risk-aware designs.
\Cref{sec:experiments} reports convergence studies and an advection--diffusion sensor-placement example comparing risk-aware goal-oriented designs against parameter-based OED.
Concluding remarks are presented in~\Cref{sec:conclusion}.

    \section{Problem Formulation}\label{sec:problem_formulation}
    OED chooses a design so that the resulting data are most informative about the model parameters or about a prediction expressed as a QoI.
    This section reviews the relevant background and introduces the linear-Gaussian example used throughout.
    
    \subsection{Bayesian Inference}\label{sec:bayes_inverse}
    Bayesian inference estimates model parameters \(\params \in \paramspace \subset \R^{\nparams}\) from data \(\data \in \obsspace \subset \R^{\nobs}\) collected under a design \(\design \in \desspace\), where \(\desspace\) denotes the set of admissible designs.
    To do so, it combines a prior density \(p(\params)\), encoding parameter uncertainty before data are observed, with a likelihood \(p(\data \mid \params, \design)\), the probability of observing \(\data\) for a parameter realization \(\params\) under design \(\design\).
    Bayes' rule then defines the posterior density of \(\params\) given \(\data\) and \(\design\) as
    \begin{equation}\label{eq:bayes_rule}
        p(\params \mid \data, \design)
        =
        \frac{p(\data \mid \params, \design)\,p(\params)}{p(\data \mid \design)}.
    \end{equation}
    The denominator is the \emph{evidence} (or marginal likelihood), \(p(\data \mid \design)=\int_{\paramspace} p(\data \mid \params, \design)\,p(\params)\,\mathrm{d}\params\).
    A common way to specify the likelihood is through an observation model relating the data to the parameters.
    The canonical example is the additive model \(\data = \fwdmodel(\params,\design) + \noisevec\), where \(\fwdmodel: \R^{\nparams} \times \desspace \to \R^{\nobs}\) is the parameter-to-observable map predicting \(\data\) from \(\params\) under design \(\design\), and \(\noisevec\) is mean-zero Gaussian measurement noise with covariance \(\noisecov\).
    The likelihood \(p(\data\mid\params,\design)\) is then the Gaussian distribution \(\mathcal{N}(\data;\fwdmodel(\params,\design),\noisecov)\).
    Evaluating the likelihood requires evaluating the map \(\fwdmodel\), which is expensive when governed by a PDE, and results in an expensive-to-compute evidence term.
    
    \subsection{Goal-Oriented Bayesian Inference}\label{sec:goal_bayes}
    Often the inference of parameters is not the goal of a study, but an intermediate step toward predicting a QoI that cannot be observed directly, such as contaminant concentration at a downstream receptor or peak stress in a mechanical component.
    Typically, these predictions are collected in a vector-valued QoI map \(\qoivec(\params) = (q_1(\params),\ldots,q_{\nqoi}(\params)) \in \qoispace \subset \R^{\nqoi}\), whose components \(q_j(\params)\) are the individual scalar predictions for \(j = 1,\ldots,\nqoi\).
    These components may be a finite collection of distinct predictions or the discretization of a predictive field at \(\nqoi\) locations.
    A parameter density \(\pi(\params)\) then induces a density over the QoI, the push-forward of \(\pi\) through \(\qoivec\),
    \begin{equation}\label{eq:post_pushforward}
        \pi_{\text{pf}}({\bm z})
        =
        \int_{\paramspace}
        \delta\!\left({\bm z}-\qoivec(\params)\right)\,
        \pi(\params)\,
        \mathrm{d}\params ,
    \end{equation}
    where \(\delta(\cdot)\) denotes the Dirac delta.
    Taking \(\pi\) to be the prior \(p(\params)\) or the posterior \(p(\params \mid \data, \design)\) yields the prior or posterior push-forward, respectively.
    In practice, push-forward distributions are commonly approximated by Monte Carlo propagation, in which samples \(\{\params^{(i)}\}\sim \pi(\params)\) are mapped through \(\qoivec\) to produce QoI samples \(\{\qoivec(\params^{(i)})\}\) distributed according to \(\pi_{\text{pf}}\).
    Because the posterior depends on the design \(\design\), so too does its push-forward: the design controls how much the data reduce predictive uncertainty in the QoI, and goal-oriented design exploits this dependence to target the QoI directly.

    \subsection{Optimal Experimental Design}\label{sec:oed_problem}
OED requires a \emph{design objective} \(\mathcal{U}(\design)\) that quantifies the expected quality of the experimental outcome under a design \(\design\), together with a feasible design space \(\desspace\) over which to optimize it.
A design can be parameterized in many ways: one may optimize sensor locations or measurement times continuously over a domain \(\descondspace\), giving \(\desspace = \descondspace^{S}\) for \(S\) measurements.
However, such feasible sets are typically nonconvex, hard to optimize, and require fixing \(S\) in advance.
We instead adopt a relaxation over a fixed set of candidate conditions that yields a convex feasible set amenable to gradient-based optimization.

Let \(\vdescond=\{\descond_1,\ldots,\descond_\nobs\}\subset\descondspace\) be \(\nobs\) candidate experimental conditions, for example sensor locations, measurement times, or controlled inputs.
We represent a design by a nonnegative weight vector \(\design=(w_1,\ldots,w_{\nobs})^\top\) on the simplex,
\[
    \simplex := \left\{ \design\in\R^{\nobs} \;\middle|\; w_k \ge 0\ \forall k,\ \ \sum_{k=1}^{\nobs}w_k = 1 \right\},
\]
which we take as the design space \(\desspace\) for this parameterization.
Here, \(w_k\) is the fraction of the total measurement budget allocated to condition \(\descond_k\), so the weights act as continuous proxies for the relative importance of each candidate.
A design with a fixed number of measurements is recovered by rounding the optimized weights to integer allocations, for which standard heuristics exist~\cite{YU201844}.

Because the data \(\data\) are unknown at design time, the design objective must account for their distribution \(p(\data\mid\design)\); conventional objectives do so by taking an expectation over \(\data\sim p(\data\mid\design)\).
We frame the objective \(\mathcal{U}(\design)\) as the predictive uncertainty remaining under a design and minimize it, so that smaller values indicate better designs:
\begin{equation}\label{eq:OED_optimization}
    \design^\star \in \argmin_{\design \in \simplex} \mathcal{U}(\design).
\end{equation}

A canonical Bayesian objective for parameter-focused OED is the expected information gain (EIG), defined as the expected Kullback--Leibler (KL) divergence from the posterior to the prior:
\begin{equation}\label{eq:OED_obj}
    \EIG(\design)
    :=
    \mathbb{E}_{p(\data\mid\design)}
    \!\left[
        \phi_{\mathrm{KL}}\!\left(p(\params \mid \data, \design)\,\|\,p(\params)\right)
    \right]
    =
    \int_\obsspace \int_{\paramspace}
        \ln\!\left(
            \frac{p(\params \mid \data, \design)}{p(\params)}
        \right)
        p(\params \mid \data, \design)\,
        p(\data \mid \design)\,
    \mathrm{d}\params\,\mathrm{d}\data .
\end{equation}
EIG measures the expected reduction in parameter uncertainty, so more informative designs have larger EIG.
Because our convention instead penalizes the uncertainty that remains, we minimize its negation.
Specifically, we set the KL design objective
\begin{equation}\label{eq:kl_objective}
  \kluty := -\EIG(\design).
\end{equation}
Minimizing \(\kluty\) selects designs whose posteriors, averaged over likely datasets \(\data\sim p(\data\mid\design)\), differ as much as possible from the prior in the KL sense.
This is natural for parameter estimation, but does not distinguish parameter directions according to their influence on the QoI map \(\qoivec\).
This limitation motivates the goal-oriented and risk-aware objectives of \secref{sec:utility}, against which we compare EIG as a baseline throughout.
    
    \subsection{Linear Gaussian Example}\label{sec:linear_gaussian_running_example}
    To build intuition and enable analytical validation of our computational framework, we introduce a concrete running example that we use throughout the paper.
    We consider a linear observation model \(\obsmapeval=\obsmtx\params\) with additive Gaussian noise,
    \begin{equation}\label{eq:obs_model}
        \data = \obsmtx\params + \noisevec,
        \qquad
        \noisevec \mid \design \sim \mathcal{N}\!\left(\vv{0},\,\effnoisecov\right).
    \end{equation}
    Here \(\obsmtx\in\R^{\nobs\times\nparams}\) is the observation matrix, with each row determined by a fixed candidate location, so \(\obsspace\subseteq\R^{\nobs}\).
    The design enters only through the noise.
    Specifically, we define a baseline noise covariance that is independent across candidates, \(\noisecov = \diag(\sigma_{\varepsilon,1}^2,\ldots,\sigma_{\varepsilon,\nobs}^2)\), and let the weight \(\designk{k}\) scale the precision at candidate \(k\), giving the diagonal effective noise covariance
    \[
        \effnoisecov = \diag\!\left(
            \frac{\sigma_{\varepsilon,1}^2}{\designk{1}},\;\ldots,\;
            \frac{\sigma_{\varepsilon,\nobs}^2}{\designk{\nobs}}
        \right).
    \]
    A larger weight shrinks the effective variance at that candidate.
\iflongappendix
    With a Gaussian prior \(\params\sim\mathcal{N}(\priormean,\priorcov)\), the model is conjugate: the posterior is Gaussian with design-dependent covariance \(\postcov\) and data-dependent mean \(\postmean\), given in closed form in \appref{app:model}.
\fi
    
    To illustrate goal-oriented design with non-Gaussian predictive distributions, we define the QoI map as a componentwise exponential of a linear functional,
    \begin{equation}
     \label{eq:lognormal_qoi}
        \qoivec(\params) = \exp\!\left(\predmtx^\top \params\right) \in \R^{\nqoi},
    \end{equation}
    where \(\predmtx=[\predvec_1,\ldots,\predvec_{\nqoi}]\in\R^{\nparams\times\nqoi}\) and each \(\predvec_j\in\R^{\nparams}\) specifies a prediction direction.
    The components index a prediction-space variable: each direction \(\predvec_j = \predvec(x_j)\) is associated with a location \(x_j\in\qoidomain\), so neighboring components correspond to nearby prediction locations.
    Because \(\predmtx^\top \params\) is Gaussian under both the prior and posterior, the corresponding prior and posterior push-forward distributions of \(\qoivec(\params)\) are lognormal.
    Thus, even though the parameter posterior is Gaussian, the push-forward is non-Gaussian, making this a nontrivial yet analytically tractable benchmark with closed-form expressions for the risk-aware design objectives introduced in \secref{sec:utility}.

    This example also exposes the central limitation of parameter-focused design.
    \figref{fig:pushforward} contrasts two designs: \(\design_A\) aligns its sensor with the observation model and sharply collapses the posterior along \(\theta_1\), while \(\design_B\) aligns with the QoI direction \(\predvec\) and collapses it along \(\theta_2\) instead.
    Although \(\design_A\) produces the more concentrated posterior, its posterior push-forward is similar to the prior push-forward, whereas \(\design_B\) sharply reduces predictive uncertainty (\figref{fig:pushforward}b).
    The two designs therefore disagree across uncertainty measures: \(\design_A\) dominates the parameter-space criterion (posterior covariance determinant, \figref{fig:pushforward}c), while \(\design_B\) dominates the push-forward standard deviation (\figref{fig:pushforward}d).
    A design can thus be highly informative of the model parameters, while failing to inform predictions, with an extreme case given by a sensor placement that only informs parameter dimensions orthogonal to the QoI.
    This motivates the goal-oriented objectives developed in \secref{sec:utility}.

 \begin{figure}[!htbp]
\centering
\includegraphics[width=\textwidth]{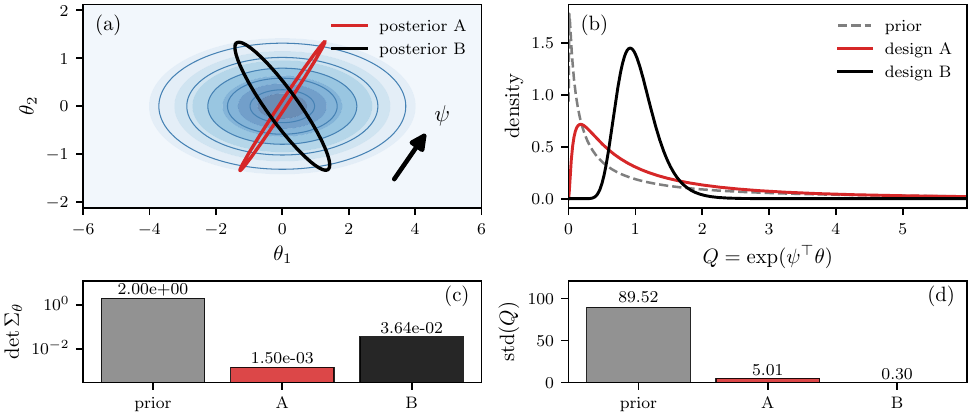}
\caption{%
 Observation model, QoI map, and posterior push-forward on the
 lognormal running example (\secref{sec:linear_gaussian_running_example}).
 \textit{(a)} Parameter space with anisotropic prior (blue filled
 contours) and $2\sigma$ posterior ellipses under designs $\design_A$
 (red) and $\design_B$ (black); the arrow marks the QoI direction
 $\predvec$.
 \textit{(b)} Prior (gray dashed) and posterior push-forwards of
 $\qoimap(\params)=\exp(\predvec\T\params)$ under $\design_A$ (red,
 nearly identical to the prior) and $\design_B$ (black, substantially
 narrower).
 \textit{(c, d)} Posterior covariance determinant $\det(\postcov)$ and
 push-forward standard deviation $\Stdunder{p(\params\mid\data,\design)}{q}$
 for the prior, $\design_A$, and $\design_B$.%
}
      \label{fig:pushforward}
    \end{figure}

\section{Risk-Aware Goal-Oriented Design Objective}
\label{sec:utility}

A risk-aware goal-oriented design objective must turn a posterior
push-forward into a single scalar that can be minimized over candidate
designs.
Formulating the scalar objective requires three successive aggregation steps over: (i) the posterior push-forward uncertainty in the QoI given a specific data realization; (ii) the QoI which may be a field on a prediction domain; and (iii) and the pre-experimental uncertainty about what data will data realization will be observed, captured by the marginal \(p(\data\mid\design)\).
Each aggregation must explicitly encodes the practitioner's attitude toward risk.
This section develops a design-objective construction that makes
these aggregation choices clear and separates them into three
independent stages.
Level~1 summarizes the spread of the posterior push-forward for a
fixed dataset; level~2 aggregates these spreads across the components
of a vector- or field-valued QoI; and level~3 aggregates across the
data not yet observed.
Level~1 uses a \emph{deviation measure}; levels~2 and~3 use a
\emph{risk measure}, each tied to practitioner risk preferences
through the risk-quadrangle framework~\cite{Rockafellar_U_SORMS_2013}.

\subsection{Deviation and Risk Measures}
\label{sec:utility:framework}

A \emph{deviation measure} $\devmeas$ quantifies how non-constant
a random variable $Z$ is: it is non-negative, vanishes if and
only if $Z$ is almost surely constant, and is insensitive to
location shifts, $\devmeas[Z + c] = \devmeas[Z]$ for any
constant $c \in \R$.
It is often generated from a regret measure $\mathcal{V}$ that quantifies a practicioner's displeasure. 
Specifically, the deviation is the minimum expected error after optimal recentering,
\begin{equation}
  \label{eq:dev:from:regret}
    \mathcal{D}(Z) =  \inf_{d\in\R}\,\{d+\mathcal{V}(Z-d)\}-\E{Z}.
\end{equation}

Certain deviation measures can be generated by a \emph{regret function} $v : \R \to \R$ and its associated \emph{error function} $e(x)=v(x)-x$, which quantitatively encode the practitioner's aversion to large, high-consequence events.
The shape of $v$ sets how spread is penalized: a symmetric regret
penalizes departures in either direction equally, like the standard
deviation, suiting cases where both directions matter to the
practitioner.
An asymmetric regret weights one tail more heavily, suiting cases where only one direction matters, such as a contaminant concentration exceeding a regulatory threshold (see \figref{fig:regret} for example regret functions).

A \emph{risk measure} $\mathcal{R}$~\cite{Rockafellar_U_SORMS_2013}
provides a quantitative answer to the question ``is the random
outcome $Z$ acceptable relative to a threshold $\tau$, given that
uncertainty may occasionally send it above $\tau$?'', producing a
scalar summary $\mathcal{R}[Z]$ that reflects both the
distribution's location and the weight of its unfavorable outcomes.
We focus on risk measures that decompose into a mean and the
associated deviation so the same regret function fixes both operators at once, i.e.,
\begin{equation}
  \label{eq:risk:decomp}
  \mathcal{R}(Z) = \E{Z} + \mathcal{D}(Z) = \inf_{d\in\R}\,\{d+\E{v(Z-d)}\}.
\end{equation}

Our framework accepts any valid deviation measure at level~1 and any
valid risk measure at levels~2 and~3.
In this paper, we use three regret measure families as representative
examples (\Cref{tab:families}): a safety-margin family, AVaR, and
entropic risk.
\begin{table}[h!]
\centering
\caption{Regret, deviation, and risk measure families.}
\begin{tabular}{lccc}
\toprule
 & Regret Measure & Deviation measure & Risk measure \\
\midrule
(safety-margin)
& $\mathcal{V}(Z) = c\lVert Z \rVert_2 + \E{Z}$
& $\mathcal{D}(Z) = c\,\Std{Z}$
& $\mathcal{R}(Z) = \E{Z} + c\,\Std{Z}$ \\[6pt]

(AVaR)
& $\mathcal{V}(Z) = \E{\tfrac{\max(0,Z)}{1-\alpha}}$
& $\mathcal{D}(Z) = \AVaR{\alpha}{Z - \E{Z}}$
& $\mathcal{R}(Z) = \AVaR{\alpha}{Z}$ \\[6pt]

(entropic)
& $\mathcal{V}(Z) = \E{\tfrac{1}{\lambda}\big(e^{\lambda Z} - 1\big)}$
& $\mathcal{D}(Z) = \tfrac{1}{\lambda}\log\E{e^{\lambda(Z-\mathbb{E}[Z])}}$
& $\mathcal{R}(Z) = \tfrac{1}{\lambda}\log\E{e^{\lambda Z}}$ \\
\bottomrule
\end{tabular}
\label{tab:families}
\end{table}
The safety-margin family is symmetric, penalizing deviations above and below the mean in the same way;\footnote{%
The safety-margin deviation measure that yields the mean-plus-standard-deviation risk measure (used in this paper; see Table~\ref{tab:families}) cannot be written as the expectation of a regret function. However, the mean-plus-variance variant of the safety-margin family can be expressed as the expectation of the regret function \(v(t)=t+ct^2\).%
}
AVaR ignores deviations below the
$(1-\alpha)$-quantile and averages the upper tail; and the entropic
family reweights the distribution exponentially toward the upper
tail.
The associated regret functions and their action on a reference distribution are
depicted in \figref{fig:regret} and \figref{fig:deviation:dist}.
The AVaR, in terms of the $\VaR{\alpha}{Z} = \inf\{z\in\R\, \vert\, \mathbb{P}(Z \leq z) \geq \alpha\}$, is defined as
\begin{equation}
  \label{eq:risk:avar}
  \mathcal{R}^{\mathrm{AVaR}}_\alpha(Z)
    = \VaR{\alpha}{Z}
      + \frac{1}{1-\alpha}\,\E{\max\!\left(0,\, Z - \VaR{\alpha}{Z}\right)} = \frac{1}{1-\alpha}\int_{\alpha}^1\VaR{t}{Z}\,\text{d}t.
\end{equation}

\subsection{The Design Objective}
\label{sec:utility:objective}

The design objective composes three stages, each applying a deviation or risk measure to one source of uncertainty in the posterior push-forward.
For simplicity, we use the operator's subscript to denote the distribution or measure it acts under, and its bracketed argument is the random variable being summarized (the conditioning is not repeated inside the bracket).
For a fixed dataset, level~1 applies the deviation measure to the
posterior push-forward of each scalar QoI component $q_j(\params)$.
We write this push-forward deviation as $\devmeas$ and the resulting
level-1 output as
\begin{equation}
  \label{eq:deviation}
  \devk{j} \coloneqq \devmeasp\!\left[q_j(\params)\right],
\end{equation}
the deviation of the posterior push-forward of component
$q_j(\params)$.
For the same fixed dataset, level~2 aggregates these deviations
across the QoI components with a risk measure under a probability
measure $\mu$ on the prediction domain~$\qoidomain$, giving one
scalar,
\begin{equation}
  \label{eq:level2}
  \levtwo \coloneqq \qoiriskm\!\left[
    \devk{1}, \ldots, \devk{\nqoi}
  \right].
\end{equation}

The measure $\mu$ sets the relative importance of each prediction
location: for a collection of scalar QoIs, it is a user-assigned
importance weighting (a discrete probability mass over the
components), and for a discretized field it is the quadrature rule
approximating integration over $\qoidomain$.
In both cases $\mu$ is represented computationally by a vector of
prediction weights $\{\gamma_j\}_{j=1}^{\nqoi}$ with $\gamma_j \geq 0$
and $\sum_j \gamma_j = 1$, which discretise $\mu$ onto the $\nqoi$
receptor locations; uniform weights $\gamma_j = 1/\nqoi$ recover
the mean over components.

Since the data are unobserved at design time, level~3 aggregates
$\levtwo$ across the data marginal $p(\data\mid\design)$.
Composing the three stages gives the complete risk-aware
goal-oriented design objective:
\begin{equation}
  \label{eq:utility:full}
  \utility =
  \underbrace{
    \datariskd
  }_{\text{level 3: data risk}}
  \!\!\left[\;
  \underbrace{
    \qoiriskm
  }_{\text{level 2: QoI risk}}
  \!\!\left[\;
  \underbrace{
    \left(\devmeasp\!\left[
      q_j(\params)
    \right]\right)_{j=1}^{\nqoi}
  }_{\text{level 1: deviation}}
  \;\right]\;\right],
\end{equation}
which we minimize over $\design \in \simplex$.
Because each stage is built from a deviation or risk measure,
$\utility$ is non-negative and vanishes only in the absence of
posterior uncertainty.
The next three subsections treat each stage in turn, providing its
operator, its action on the running example, and the elicitation
question it poses.

\begin{figure}[!htbp]
  \centering
   \includegraphics[width=\textwidth]{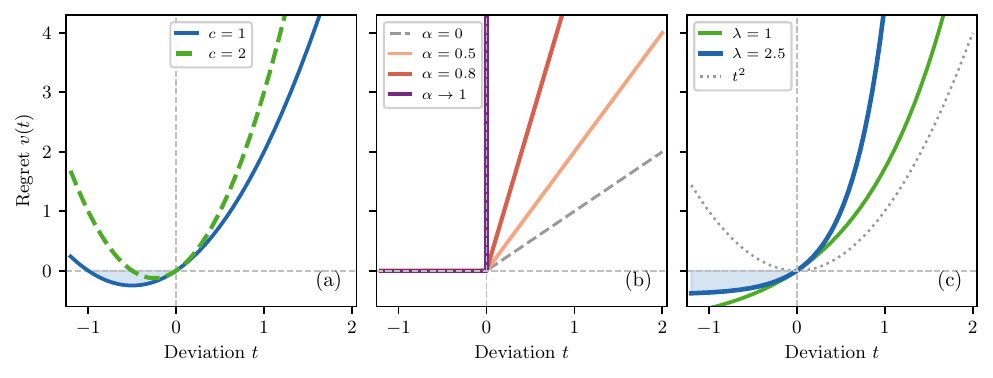}
  \caption{%
    Regret functions $v(t)$ for the three families of
    \Cref{tab:families}, plotted over
    $t\in[-1.2,\,2]$ so that asymmetry is visible.
    \textit{(a)} linear symmetric regret $v(t) = t+ct^2$
    (the variance form of the safety-margin family) at two values of $c$; both signs
    penalized equally.
    \textit{(b)} piecewise-linear regret (AVaR); negative
    deviations carry zero penalty, positive deviations are
    penalized with slope $1/(1-\alpha)$, diverging as
    $\alpha\to 1$.
    \textit{(c)} exponential regret
    $v(t) = (e^{\lambda t}-1)/\lambda$ (entropic) at two values
    of $\lambda$; the regret is \emph{negative} for $t<0$
    (shaded), so below-mean outcomes reduce the deviation while
    above-mean outcomes are penalized super-quadratically.%
  }
  \label{fig:regret}
\end{figure}

\begin{figure}[!htbp]
  \centering
  \includegraphics[width=\textwidth]{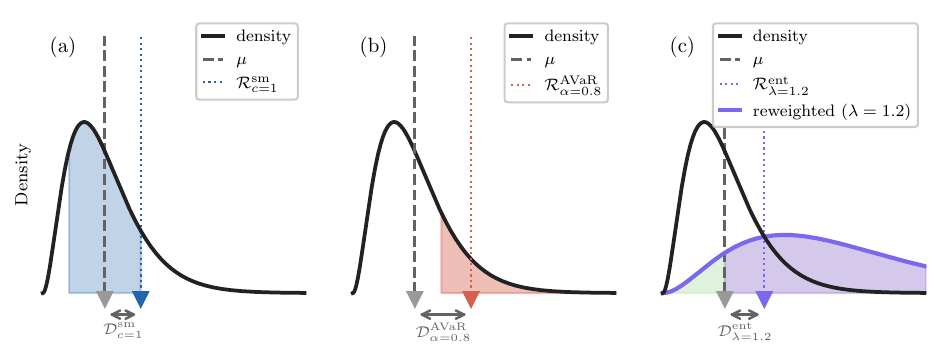}
  \caption{%
    How each deviation measure operates on a reference distribution
    of Level-1 outputs $\devkb{j}$ (black curve;
    right-skewed, as is typical for push-forward deviation values).
    In each panel the gray dashed line marks the mean $\mu$; the
    colored dotted line marks the risk measure value
    $\mathcal{R}$; the gray double-headed arrow shows the
    deviation $\mathcal{D} = \mathcal{R} - \mu$.
    \textit{(a)} Standard deviation ($c=1$): the shaded region marks
    $\pm 1\sigma$ around the mean.
    \textit{(b)} AVaR deviation ($\alpha=0.8$): the shaded tail
    region above the $\alpha$-quantile defines the conditional mean
    $\mathcal{R}^{\mathrm{AVaR}}$.
    \textit{(c)} Entropic deviation ($\lambda=1.2$): the purple
    curve is the exponentially reweighted density; below-mean
    outcomes (light green) carry weight $<1$ and reduce the
    deviation, above-mean outcomes (light purple) carry weight $>1$
    and are amplified.%
  }
  \label{fig:deviation:dist}
\end{figure}

\subsection{Level~1: Deviation of the Posterior Push-forward}
\label{sec:utility:level1}

Level~1 applies a deviation measure to the posterior push-forward of
a scalar QoI component, producing the deviation $\devk{j}$
of \eqref{eq:deviation}.
The push-forward, and hence $\devkb{j}$, depends on the observed data
$\data$, and captures only the spread of the push-forward, not its
mean.
On the running example (\secref{sec:linear_gaussian_running_example}), the posterior push-forward of each component
$q_j(\params) = \exp(\predvec_j\T\params)$ is a univariate lognormal,
and applying the safety-margin deviation with $c=1$ (the standard
deviation) returns one scalar per dataset.
Because the push-forward changes with the data, this scalar varies
across realizations, producing the deviation surface in
\figref{fig:level1}.
Since the prediction directions $\predvec_j = \predvec(x_j)$ vary
continuously with the prediction location $x_j$
(\secref{sec:linear_gaussian_running_example}), the deviation varies
smoothly across components, so the surface appears smooth even though
it is assembled from finitely many components.

\begin{figure}[!htbp]
  \centering
  \includegraphics[width=\textwidth]{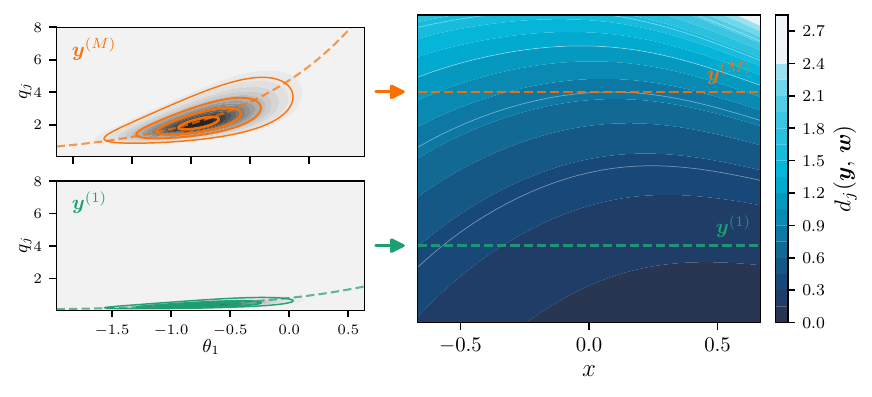}
  \caption{%
    Depiction of level-1 standard deviation measure for the running example. 
    \textit{Left:} joint posterior densities $p(q_j,\theta\mid \datam{m},\design)$
    for two realizations of data $\datam{1}$ and
    $\datam{\nout}$ concentrating along the QoI curve.
    \textit{Right:} the standard deviation surface interpolated over 
    the components of the QoI ($x$-axis) and data realizations ($y$-axis) 
    where the two data realizations are noted with horizontal lines.
  }
  \label{fig:level1}
\end{figure}

The choice of deviation measure encodes how the practitioner values
spread and changes the surface in the right panel of
\figref{fig:level1}.
The standard deviation, being symmetric, suits cases where departures
in either direction matter equally; when only one direction carries
the consequence---for example, a contaminant concentration that is dangerous
only when it exceeds a threshold---the AVaR or entropic deviation,
which weight the upper tail more heavily, are preferable.
Minimizing the resulting objective therefore selects the design that
most reduces spread, weighting departures in each direction
according to the practitioner's preferences.

\subsection{Level~2: Aggregation Across QoI Components}
\label{sec:utility:level2}

For a vector- or field-valued QoI
$\qoivec = (q_1(\params), \ldots, q_{\nqoi}(\params))$, level~1
produces one deviation per component, which level~2 aggregates into
$\levtwo$ of \eqref{eq:level2} using the prediction-domain
risk measure \(\qoiriskm\).
On the running example, level~2 collapses the prediction-location
axis of the deviation surface to a curve $\levtwo$, one
scalar per dataset.
\figref{fig:level2} shows this curve under two choices of risk
measure on the same surface.
Both rise with the data, but summarize the deviation profile
differently: $\mathrm{AVaR}_\alpha$ averages only the
worst-performing fraction of prediction locations, while the
entropic measure reweights the entire profile exponentially toward
its upper tail.

\begin{figure}[!htbp]
  \centering
  \includegraphics[width=\textwidth]{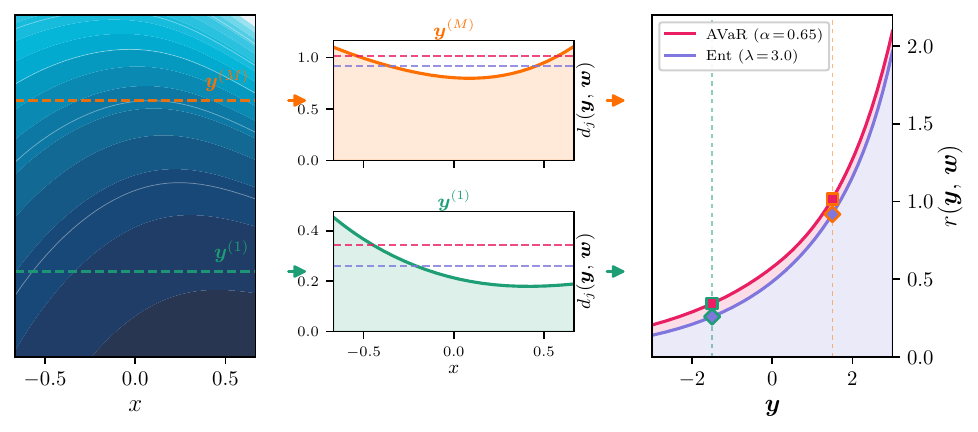}
  \caption{%
    Level-2 QoI risk measure: aggregating the deviation $\devkb{j}$
    over the QoI components at a fixed design $\design$.
    \textit{Left:} the deviation surface of \figref{fig:level1}, with
    two data realizations marked (dashed horizontal lines).
    \textit{Middle:} the deviation profile $\devkb{j}$ for each of the two data realizations, 
    with horizontal lines marking
    the scalar $\levtwo$ assigned by $\mathrm{AVaR}_\alpha$ and entropic
    risk measures.
    \textit{Right:} the resulting $\levtwo$ as a function of the
    data $\data$, with the two realizations marked.
  }
  \label{fig:level2}
\end{figure}

Specifying $\alpha\in[0,1)$ in level~2 provides a tunable mechanism for controlling the relative emphasis placed on average versus extreme prediction uncertainties: $I$-optimality ($\alpha=0$) recovers spatially averaged uncertainties, while increasing $\alpha$ more heavily penalizes extreme uncertainties, which can be highly relevant if there exist critical locations, such as a downstream population center, or a sensor near sensitive infrastructure, where where poor predictive performance carries disproportionately high consequences.
As $\alpha$ ranges from $0$ to $1$, the objective sweeps from the
average to the single worst component.
With the standard deviation at level~1, this is the classical
interpolation between $I$- and $G$-optimality~\cite{Kouri_JH_SIAMUQ_2021}.

\subsection{Level~3: Aggregation Across Data Realizations}
\label{sec:utility:level3}

Since optimal designs are determined prior to collecting data $\data$, the output of level-2 aggregation $\levtwo$ is
itself random under the marginal $p(\data\mid\design)$.
Level~3 aggregates it with a risk measure to give the scalar design
objective,
\begin{equation}
  \label{eq:level3}
  \utilityb(\design) \coloneqq \datariskd\!\left[\levtwo\right].
\end{equation}
On the running example, the induced distribution of $\levtwob$ is
right-skewed: most datasets give moderate deviation, but a small
fraction of adverse outcomes give disproportionately large $\levtwob$.
\figref{fig:level3} shows two level-3 choices on this distribution:
the mean weights all datasets equally, while $\mathrm{AVaR}_\beta$
concentrates on the right tail.

\begin{figure}[!htbp]
  \centering
  \includegraphics[width=0.5\textwidth]{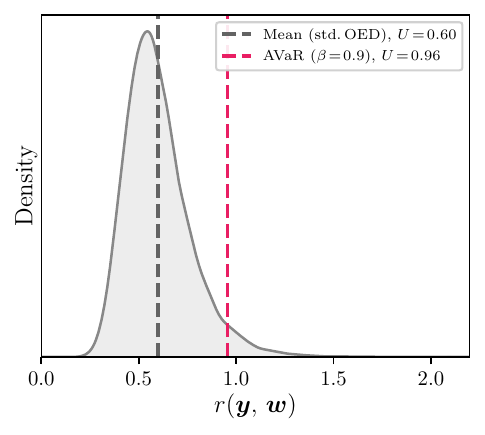}
  \caption{%
   Application of Level-3 data risk measures which acts on the probability density of the
   Level-2 aggregation $\levtwob$ where $y \sim p(\data\mid\design)$. standard Bayesian OED objective takes the average (risk-neutral) 
   and $\mathrm{AVaR}_\beta$ ($\beta=0.9$) weights undesirable large values more heavily than small values.
  }
  \label{fig:level3}
\end{figure}

The choice of Level~3 measure answers how much the practitioner is
willing to be exposed to the adverse datasets that may be realized.
If being right on average across possible datasets is acceptable,
the standard Bayesian OED approach that leverages the mean is the natural choice.
However, if the practitioner wants to guard against adverse datasets,
even at the cost of average performance,
$\mathrm{AVaR}_\beta$ averages $\levtwo$ over only the
worst $(1-\beta)$-fraction of datasets, so the objective measures
the uncertainty left under the most adverse outcomes.
Minimizing it therefore selects the design that most reduces
uncertainty under those adverse datasets, rather than on average.
In contrast, the safety margin $\mathcal{R}^{\mathrm{sm}}_c$ offers an
alternative with an engineering interpretation: the design is
acceptable with high probability when $c$ covers a chosen number of
standard deviations above the mean, with larger $c$ being more
conservative.

\paragraph{When preferences cannot be agreed upon.}
In practice multiple stakeholders may hold competing risk
preferences, or a single stakeholder may be genuinely uncertain
about their risk attitude.
Rather than commit to a single specification, one can examine how
sensitive the optimal design is to the choice of deviation and risk
measures, and so identify which choices materially affect the design.
We demonstrate this in \secref{sec:experiments:advdiff:cost} using a
cost-effective approach that quantifies how much each level and
risk measure changes the recommended experiment.

\section{Estimation and Optimization}
\label{sec:computation}

The design objective $\utility = \datariskd\!
\bigl[\qoiriskm\bigl[\devmeasp
[\qoi]\bigr]\bigr]$ nests three operations: an inner posterior
deviation $\devmeas$ over $p(\params\mid\data,\design)$ for each
QoI component (level~1), a risk measure $\qoirisk$ aggregating
across QoI components (level~2), and an outer risk measure
$\datarisk$ over the data marginal $p(\data\mid\design)$
(level~3).
Levels~1 and~3 are expectations against distributions that are
not available in closed form for general nonlinear models and
priors, and so must be approximated by sampling; level~2, in
contrast, acts on the already-sampled per-component deviations
and requires no additional forward-model evaluations.
This section develops the computational machinery for estimating
and minimizing $\utility$ on this basis.
Although the three-level framework of \secref{sec:utility} places
no structural restrictions on the likelihood or the design space,
the estimators and gradient expressions developed below
specialize to the weighted Gaussian likelihood and simplex design
parameterization of \secref{sec:linear_gaussian_running_example}; analogous derivations
for other Gaussian observation models or for continuous design
spaces are possible but not pursued here.

\subsection{Double-Loop Estimator}
\label{sec:computation:quadrature}

Evaluating the risk-aware objective requires a quadrature rule at
each of its three levels.
We write each with general nodes and weights, so that any
quadrature strategy, e.g., Monte Carlo, quasi-Monte Carlo, sparse
grids etc., can be substituted without modifying the algorithm.
We focus on Levels~1 and~3, since the Level-2 rule is simply the
application of the practitioner-supplied nodes and weights as a
weighted sum over the prediction domain.

\paragraph{Outer quadrature: data realizations.}
To generate likely data $\data \sim (\data | \design)$, we generate plausible datasets from the model itself.
Specifically, we draw parameters from the prior, push them through
the forward model, and add measurement noise from the weighted
Gaussian observation likelihood of \secref{sec:bayes_inverse}.
In this model the likelihood is determined by the noise covariance,
and the design enters by allocating measurement budget across the
candidates.
The baseline covariance is diagonal,
$\noisecov = \diag(\sigma_{\varepsilon,1}^2,\ldots,\sigma_{\varepsilon,\nobs}^2)$,
and the weight $\designk{k}$ scales the precision at candidate $k$,
so the effective covariance is
$\effnoisecov = \diag(\sigma_{\varepsilon,k}^2/\designk{k})$,
as in the running example~\eqnref{eq:obs_model}.
The resulting log-likelihood is
\begin{equation}
  \label{eq:likelihood}
  \log \likeli
  \propto
  -\frac{1}{2}
  \sum_{k=1}^{\nobs}
  \frac{\designk{k}}{\sigma_{\varepsilon,k}^2}
  \bigl(y_k - \obsmap_k(\params,\design)\bigr)^2.
\end{equation}
A larger weight $\designk{k}$ sharpens the $k$-th term and
$\designk{k}=0$ removes candidate $k$ from the experiment.

To enable gradient-based optimization of the design, we generate
the data with a reparameterization trick.
We take outer quadrature nodes and weights
$\{(\outersamp{m}, \outwt{m})\}_{m=1}^{\nout}$, drawing each
$\outersamp{m}$ as a joint sample from $p(\params, \noisevec)$, and
form each data realization as
\begin{equation}
  \label{eq:reparam}
  \datam{m}_k
  = \obsmap_k(\outersamp{m}, \design)
  + \frac{\sigma_{\varepsilon,k}}{\sqrt{\designk{k}}}\,\noisem{m}_k,
  \qquad k = 1,\ldots,\nobs,
\end{equation}
holding $\noisem{m} \sim \Gaussian{\vv{0}}{\mm{I}}$ fixed as
$\design$ varies.
This moves the design dependence out of the sampling distribution
and into a deterministic transformation: once $\outersamp{m}$ and
$\noisem{m}$ are fixed, $\datam{m}(\design)$ is a differentiable
function of $\design$, enabling gradient computation.
In subsequent expressions, we suppress the $\design$ argument and
write simply $\datam{m}$, restoring it only when the design
dependence is the focus (as in \secref{sec:computation:gradients}).

\paragraph{Inner quadrature: posterior via importance weights.}
The level-1 deviation is an expectation under the posterior
$p(\params\mid\datam{m},\design)$, which differs for every outer
dataset $\datam{m}$.
Resampling the posterior from scratch for each $\datam{m}$ would
make the inner cost scale with the number of outer samples, the
expense that makes naive double-loop estimators prohibitive.
To avoid his, we draw a single set of prior samples and
re-weight them toward each dataset's posterior with importance
weights, so the same inner samples can be leveraged across all outer data samples.

To compute the level-1 deviation of a QoI component $q_j(\params)$
(\secref{sec:goal_bayes}), we draw inner quadrature nodes and
weights $\{(\innersamp{n}, \inwt{n})\}_{n=1}^{\nin}$ from the prior
$\prior$.
We use these to estimate both the normalizing constant (the
evidence) and moments of $q_j$ against the posterior.
For example, expressed as prior expectations via self-normalized
importance sampling, the posterior standard deviation is
\begin{equation}
  \label{eq:stddev_unnormalized}
  \Stdunder{\post}{q_j}
  = \left[
      \frac{\Eunder{\prior(\params)}{q_j^2(\params)\, p(\datam{m}\mid\params,\design)}}{p(\datam{m}\mid\design)}
    - \left(
        \frac{\Eunder{\prior(\params)}{q_j(\params)\, p(\datam{m}\mid\params,\design)}}{p(\datam{m}\mid\design)}
      \right)^{\!2}
    \right]^{1/2}.
\end{equation}
We compute the numerator moments and the evidence
$p(\datam{m}\mid\design)$ in the denominator with the \emph{same}
quadrature rule $\{(\innersamp{n}, \inwt{n})\}$.
This collapses~\eqnref{eq:stddev_unnormalized} into a ratio of
sums governed by self-normalized importance weights,
\begin{equation}
  \label{eq:impweights}
  \impwt{m}{n}
  = \frac{
      \inwt{n}\, p(\datam{m} \mid \innersamp{n}, \design)
    }{
      \displaystyle\sum_{n'=1}^{\nin}
      \inwt{n'}\, p(\datam{m} \mid \innersamp{n'}, \design)
    },
\end{equation}
and turns a nominally $\nin \times \nout$ sampling problem into
an $\nin + \nout$ one.
The cost is that the numerator and denominator estimates are now
correlated --- making the ratio a biased estimator of the true
posterior variance --- but the bias vanishes as $\nin \to \infty$
and the estimator remains asymptotically consistent for all
risk-aware utilities discussed in this paper.

\subsection{The Estimation Algorithm}
\label{sec:computation:algorithm}

\Algref{alg:estimator} assembles the outer data realizations and
the inner importance weights into the design-objective estimate:
it loops over data realizations, reuses the shared inner rule to
compute posterior deviations for each QoI component, and aggregates
them first across components (level~2) and finally across data
(level~3).
The operators $\devmeas$, $\qoirisk$, and $\datarisk$
corresponding to these levels are each realized by substituting
an expression from
the families of \secref{sec:utility:framework}, and the level-2 aggregation weights the QoI components by the
prediction weights $\{\gamma_j\}_{j=1}^{\nqoi}$ introduced in
\secref{sec:utility:framework} as the discretization of $\mu$.
Throughout the algorithm, the notation $\{\hat{d}_j^{(m)}, \gamma_j\}_j$
denotes an empirical weighted distribution with atoms
$\hat{d}_j^{(m)}$ and weights $\gamma_j$; similarly
$\{\hat{r}^{(m)}, \outwt{m}\}_m$ at level~3.
Each risk operator $\qoirisk$ and $\datarisk$ consumes exactly
this representation and can be swapped without touching any other
part of the algorithm.
Because all outer and inner samples are generated once and held
fixed across design evaluations, the algorithm returns a
deterministic function of $\design$, a property we exploit when
computing gradients in \secref{sec:computation:gradients}.
 
\begin{algorithm}[t]
  \SetAlgoLined
  \KwIn{%
    Design $\design \in \simplex$;
    outer quadrature
    $\{(\outersamp{m}, \noisem{m}, \outwt{m})\}_{m=1}^{\nout}$
    with $\outersamp{m}\in\R^{\nparams}$,
    $\noisem{m}\in\R^{\nobs}$;
    inner quadrature
    $\{(\innersamp{n}, \inwt{n})\}_{n=1}^{\nin}$ with
    $\innersamp{n}\in\R^{\nparams}$;
    prediction weights $\{\gamma_j\}_{j=1}^{\nqoi}$;
    level choices $\devmeas$ (deviation),
    $\qoirisk$ (level~2),
    $\datarisk$ (level~3).%
  }
  \KwOut{Design-objective estimate $\hat{U}(\design)$.}
  \BlankLine
  \tcp{Data generation (outer samples).}
  Evaluate forward model:
    $\mathbf{G}_{:,m} \leftarrow \obsmap(\outersamp{m}, \design)$
    for $m = 1,\ldots,\nout$
    \tcp*{$\mathbf{G}\in\R^{\nobs\times\nout}$}
  Noise scaling (diagonal):
    $\mathbf{E}_{k,m} \leftarrow
    (\sigma_{\varepsilon,k}/\sqrt{\designk{k}})\,\noisem{m}_k$
    for all $k,m$
    \tcp*{$\mathbf{E}\in\R^{\nobs\times\nout}$}
  Data matrix:
    $\mathbf{Y} \leftarrow \mathbf{G} + \mathbf{E}$;
    set $\datam{m} \leftarrow \mathbf{Y}_{:,m}$\;
  \BlankLine
  \tcp{Prior-space precomputation (reusable across designs).}
  Evaluate QoI:
    $\mathbf{Q}_{j,n} \leftarrow q_j(\innersamp{n})$ for all $j,n$
    \tcp{$\mathbf{Q}\in\R^{\nqoi\times\nin}$}
  \BlankLine
  \tcp{Outer loop.}
  \For{$m = 1, \ldots, \nout$}{
    Log-likelihoods:
      $\ell_{mn} \leftarrow \log p(\datam{m}\mid\innersamp{n},\design)$
      for $n=1,\ldots,\nin$\;
    Evidence estimate:
      $\hat{p}(\datam{m}\mid\design) \leftarrow
      \sum_{n=1}^{\nin} \inwt{n}\, e^{\ell_{mn}}$\;
    Importance weights:
      $\impwt{m}{n} \leftarrow
      \inwt{n}\, e^{\ell_{mn}} / \hat{p}(\datam{m}\mid\design)$
      \tcp*{shared inner rule}
    Deviations:
      $\hat{d}_j^{(m)} \leftarrow
      \devmeas\!\left[\{\impwt{m}{n}, \mathbf{Q}_{j,n}\}_n\right]$
      for $j=1,\ldots,\nqoi$\;
    Level-2 aggregation:
      $\hat{r}^{(m)} \leftarrow
      \qoirisk\!\left[\{\hat{d}_j^{(m)}, \gamma_j\}_j\right]$\;
  }
  \BlankLine
  Level-3 aggregation:
    $\hat{U} \leftarrow
    \datarisk\!\left[\{\hat{r}^{(m)}, \outwt{m}\}_m\right]$\;
  \Return{$\hat{U}$}
  \caption{Three-level risk-aware goal-oriented OED design objective.}
  \label{alg:estimator}
\end{algorithm}
 
Continuing the running example, in which the deviation measure is
the standard deviation $\devmeas = \mathrm{Std}$, substituting the
importance weights $\impwt{m}{n}$
into~\eqnref{eq:stddev_unnormalized} absorbs the evidence factors
and yields the computable per-component form
\begin{equation}
  \label{eq:inner_std}
  \hat{d}_j^{(m)}(\design)
  = \left[
      \sum_{n=1}^{\nin} \impwt{m}{n}\, q_j(\innersamp{n})^2
      -
      \left(\sum_{n=1}^{\nin} \impwt{m}{n}\, q_j(\innersamp{n})\right)^{\!2}
    \right]^{1/2},
\end{equation}
which reuses the shared weights $\impwt{m}{n}$ and cached QoI
values $q_j(\innersamp{n})$ for every component, requiring no
additional forward-model evaluations.
The evidence $\hat{p}(\datam{m}\mid\design)$ is retained as an
intermediate quantity because it is reused in the gradient
of~\secref{sec:computation:gradients}.

The design of \Algref{alg:estimator} makes changing the
practitioner's risk preference straightforward.
\iflongappendix
The $\mathrm{AVaR}$ operator is the one exception: it is non-smooth
on the raw samples and is replaced by a smoothed surrogate
$\widetilde{\mathrm{AVaR}}_{\alpha,\delta}$, obtained by adding a
penalty to the dual representation of $\mathrm{AVaR}$ \cite{kouri2020epi,kouri2022primal} and
recovering $\mathrm{AVaR}_\alpha$ as $\delta\to\infty$
(\appref{app:smooth-avar}); all other risk measures are directly
differentiable in the empirical weights.
\fi
\iflongappendix
\appref{app:convergence} provides a study verifying convergence of
the nested-quadrature estimator
against the closed-form design-objective expressions of
\appref{app:scalar}.
\fi

\subsection{Gradient Computation and Optimization}
\label{sec:computation:gradients}

Because all outer and inner samples are drawn once and held fixed
across design evaluations, \Algref{alg:estimator} defines a
deterministic approximation $\hat{U}(\design)$ of the true design
objective $U(\design)$.
Specifically, the algorithm returns a sample-average approximation
(SAA) in which $\design$ enters only through the reparameterized
data~\eqnref{eq:reparam} and the importance
weights~\eqnref{eq:impweights}.
Both are differentiable in $\design$, so $\hat{U}$ is differentiable
too, everywhere except at the non-smooth indicator inside
$\mathrm{AVaR}$, which we smooth as described below.
The design problem therefore reduces to the continuous, smooth
optimization
\begin{equation}
  \label{eq:saa_problem}
  \hat{\design}^\star
  = \argmin_{\design \in \simplex} \hat{U}(\design),
\end{equation}
which we solve with \texttt{Scipy}'s trust-region interior-point algorithm.  
Solving for $\design$ with gradients in this way avoids the
combinatorial search over candidate sensor sets used by some OED methods.
Also, note that \eqref{eq:reparam} is singular at $\design_k=0$, so we enforce $\design_k>10^{-6},\,\forall k $ during optimization.

The gradient $\nabla_{\design}\hat{U}$ is available by automatic
differentiation through \Algref{alg:estimator}.
This is the simplest route and the one we recommend for
prototyping, but each optimizer step then traverses the full
computation graph of the $\nout \times \nin$ double sum.
The gradient also admits a closed form, because $\design$ enters
the estimator in only two places: the reparameterized
data~\eqnref{eq:reparam} and the importance
weights~\eqnref{eq:impweights}.
This closed form reuses quantities already cached in the forward
pass, in particular the log-likelihoods $\ell_{mn}$, so it needs no
further forward-model evaluations and costs a constant multiple of
the objective.
\iflongappendix
\appref{app:gradients} derives this closed-form gradient, which we
use in our experiments, together with the smoothed-$\mathrm{AVaR}$
surrogate $\widetilde{\mathrm{AVaR}}_{\alpha,\delta}$, whose bias
vanishes as $\delta\to\infty$.
\fi

\section{Numerical Experiments}
\label{sec:experiments}

This section demonstrates two consequences of the framework.
First, a single choice of risk aggregator continuously
interpolates between $I$- and $G$-optimality.
Second, in a realistic PDE-based inverse problem, the resulting designs differ
qualitatively from one another and from the KL baseline.
When the observation geometry and the QoI structure are
misaligned, parameter-informative and prediction-informative
sensors often differ, so a design optimized for parameter
estimation may perform no better than random designs in terms of improving prediction uncertainty.

Although the three-level construction~\eqnref{eq:utility:full}
is general, each experiment exercises only a subset of the
design objectives listed in~\tabref{tab:named_utilities}, chosen to
reflect a range of risk preferences: from risk-neutral
($U_1$: std / mean / mean), through QoI-level and data-level
risk aversion ($U_2$--$U_4$), to a level-1 variation that
replaces standard deviation with an asymmetric deviation measure
($U_5$: entropic / mean / mean); the safety-margin variant $U_6$
(std / mean / mean$\,+\,c\,\cdot\,$Std) appears only in the
estimator-convergence study of the appendix.
The expected information gain (KL) is included as a
parameter-focused baseline; KL is not an instance
of~\eqnref{eq:utility:full} as it targets the full parameter
posterior rather than the push-forward, but it serves as a
point of comparison throughout.
 
\begin{table}[htb]
  \caption{%
    Named design objectives used in this paper.
    Each design objective $U_i$ corresponds to a specific choice of
    deviation ($\devmeas$), QoI risk ($\qoirisk$), and data risk
    ($\datarisk$) within the three-level
    framework~\eqnref{eq:utility:full}.
    The KL design objective (negative expected information gain) is listed for
    comparison but is not an instance of the framework.%
  }
  \label{tab:named_utilities}
  \centering
  \begin{tabular}{llll}
    \toprule
    Label & Level~1 ($\devmeas$) & Level~2 ($\qoirisk$) & Level~3 ($\datarisk$) \\
    \midrule
    $U_1$ & std      & mean                   & mean \\
    $U_2$ & std      & $\mathrm{AVaR}_\alpha$ & mean \\
    $U_3$ & std      & mean                   & $\mathrm{AVaR}_\beta$ \\
    $U_4$ & std      & $\mathrm{AVaR}_\alpha$ & $\mathrm{AVaR}_\beta$ \\
    $U_5$ & entropic & mean                   & mean \\
    $U_6$ & std      & mean                   & mean $+\,c\,\cdot\,$Std \\
    \midrule
    KL    & \multicolumn{3}{l}{expected information gain (parameter-focused baseline)} \\
    \bottomrule
  \end{tabular}
\end{table}

\subsection{$I$--$G$ Interpolation}
\label{sec:experiments:ig}
 
The following experiment demonstrates that the level-2 AVaR
parameter $\alpha$ provides a continuous interpolation between the
classical $I$-optimal and $G$-optimal designs, using the
$U_2$ design objective (std / $\mathrm{AVaR}_\alpha$ / mean; see
\tabref{tab:named_utilities}). 
To give the level-2 aggregation a multivariate target, we use
the lognormal running example of
\secref{sec:linear_gaussian_running_example}, extending its
QoI from a scalar to a vector by evaluating the
QoI map~\eqnref{eq:lognormal_qoi} at $\nqoi = 100$ prediction
locations placed equidistantly on $[-2/3,\, 2/3]$.
The observation model, basis, prior, and noise covariance are
specified in \appref{app:convergence}.
Weighting the prediction locations uniformly at level~2,
$\gamma_j = 1/\nqoi$, makes the two endpoints of the AVaR sweep
coincide with the classical criteria:
at $\alpha = 0$, $\mathrm{AVaR}_0$ reduces to the expectation, so
the objective averages posterior deviation across the prediction
domain and recovers $I$-optimality, while as $\alpha \to 1$ (attained here at $\alpha_{\max} = 0.95$) $\mathrm{AVaR}_\alpha$ approaches the maximum over QoI components and approximates $G$-optimality.
Sweeping $\alpha$ between these endpoints therefore traces a path from average-case to worst-case design. 
Moreover, the resulting objective is
continuous in $\alpha\in[0,1)$ when using either exact AVaR or the smoothed AVaR used in our numerical results.
These numerical results are supported theoretically.
\iflongappendix
First, in \appref{sec:app_smoothed_avar} we establish continuity of the smoothed $\mathrm{AVaR}_\alpha$ functional in $\alpha\in[0,1)$. Second, since level~3 is an expectation over data, this continuity (under a mild integrability assumption on the level-2 outputs) carries through to continuity of the full design objective in $\alpha$.
\fi
This argument extends directly to classical
$\mathrm{AVaR}_\alpha$, whose continuity in $\alpha\in[0,1)$ follows from analogous arguments as those in \appref{sec:app_smoothed_avar}.

For each $\alpha$, we solve the design problem
\begin{equation}
  \label{eq:ig:opt}
  \design^\star(\alpha)
  = \argmin_{\design \in \simplex}\,
    U_\alpha(\design),
\end{equation}
where $U_\alpha$ is the $U_2$ design objective evaluated at AVaR
level $\alpha$.
The sweep follows an ascending continuation in $\alpha$, applied
identically to the MC and exact optimizations: the smallest level
is initialized from the uniform design and each subsequent level
from the optimal design at the previous grid point.
The sweep uses $20$ equally spaced levels
$\alpha \in [0, \alpha_{\max}]$; the MC optimization uses
$\nout = \nin = 4000$ samples, while the exact optimization
evaluates the analytical objective with $10^4$ outer samples.
At each optimal design $\design^\star(\alpha)$, we evaluate three
objectives to test for continuity:
(i)~the matched objective $U_\alpha(\design^\star(\alpha))$,
evaluated at the same AVaR level the design was optimized for;
(ii)~the $I$-optimal objective $U_0(\design^\star(\alpha))$; and
(iii)~the $G$-optimal proxy
$U_{\alpha_{\max}}(\design^\star(\alpha))$.
By construction the matched and $I$-optimal curves must coincide at
$\alpha = 0$ and the matched and $G$-optimal curves at
$\alpha = \alpha_{\max}$.

\figref{fig:ig:sweep} shows how curves (i)-(iii) behave for varying values of $\alpha$.
As expected, $U_\alpha(\design^\star(\alpha))$ corresponds to the $I$-optimal objective at $\alpha = 0$ and
corresponds to the $G$-optimal proxy for $\alpha = \alpha_{\max}$.
For $0<\alpha<1$, $U_\alpha(\design^\star(\alpha))$ monotonically increases, providing a continuous nonlinear interpolation between 
the two classical optima.
Although the objective varies smoothly, the underlying design
does not.
The right panel of \figref{fig:ig:sweep} shows the optimal
weights differing significantly across the $\alpha$ sweep.
Sensor~9 provides a nonzero weight for all values of $\alpha$, while all other sensors have weights that are zero for some values of $\alpha$.
For most sensor locations, the weight varies (both continuously and discontinuously) across the $\alpha$ range,
highlighting how risk preferences can have substantive impacts on the resulting designs.
The discontinuous jumps in the weights arise because the design
objective admits multiple local minima with nearly identical
objective values but distinct weight configurations; as $\alpha$
crosses a threshold, the minimizer moves from one such basin to
another, so the design changes abruptly even though the optimal
objective value itself varies continuously.

 \begin{figure}[htb]
  \centering
  \includegraphics[width=\textwidth]{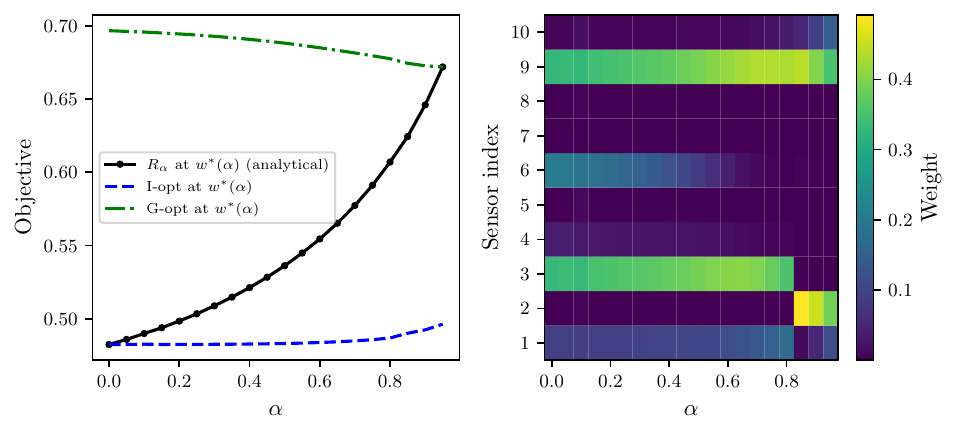}
  \caption{%
    $I$--$G$ optimality interpolation on the running examples given in~\Cref{sec:linear_gaussian_running_example}.
    \textit{Left:} design-objective values at the optimal design
    $\design^\star(\alpha)$ as a function of $\alpha$.
    Red markers: Monte Carlo estimate of the matched objective
    $U_\alpha(\design^\star(\alpha))$.
    Blue dashed: $I$-optimal objective
    $U_0(\design^\star(\alpha))$.
    Green dash-dotted: $G$-optimal proxy
    $U_{\alpha_{\max}}(\design^\star(\alpha))$.
    \textit{Right:} optimal sensor ($y$-axis) weights as a function of
    $\alpha$.
  }
  \label{fig:ig:sweep}
\end{figure}

\subsection{The Advection--Diffusion Test Problem}
\label{sec:experiments:advdiff:setup}

The remaining experiments all use a single PDE-constrained
sensor-placement problem, representative of surface-water quality
monitoring, which we describe here before turning to results.
The goal of the experiments is to predict the contaminant
concentration that will reach a set of downstream receptors when
a contaminant of unknown spatial distribution is released
upstream in a channel and transported toward them, by placing a
limited number of sensors so as to inform that prediction.

\paragraph{Forward model.}
The contaminant concentration $c(\vv{x}, t; \params)$ is
transported by advection and diffusion, evolving on the unit
square $\Omega = [0,1]^2$ according to
\begin{equation}
  \label{eq:advdiff}
  \partial_t c + \vv{u}\cdot\nabla c
  - D\,\Delta c = s(\vv{x}; \params),
\end{equation}
with source $s(\vv{x};\params)$, diffusivity $D$, Robin
conditions $D\,\partial_n c + 0.1\,c = 0$ on the left and right
boundaries, zero-flux Neumann conditions elsewhere, and zero
initial condition; time integration uses Crank--Nicolson with
$\Delta t = 0.25$ up to the final time $T = 4$.
We use an advection-dominated regime $D = 0.005$, in which the
tracer follows narrow corridors set by the velocity field.
The domain, depicted in \figref{fig:advdiff:domain}, contains
three rectangular obstructions, and the steady velocity field
$\vv{u}(\vv{x})$ that carries the contaminant from left to
right is obtained once by solving the incompressible
Navier--Stokes equations with Taylor--Hood P2/P1 finite elements
(inlet profile $u_x = [x_2(1 - x_2)]^{3/2}$, no-slip walls and
obstructions, stress-free outlet, kinematic viscosity
$\nu = 1/\mathrm{Re} = 0.08$ with $\mathrm{Re} = 12.5$).
The obstructions create recirculation zones and wakes, so the
transport pathways from the source region to the downstream
receptors are non-trivial; the resulting velocity field is shown
in the right panel of \figref{fig:advdiff:domain}.

\paragraph{Bayesian inverse problem.}
In the following experiments, we treat the source as uncertain and
infer it from sensor data.
We model the source as a lognormal random field
$s(\vv{x};\params) =
\exp(\sum_{k=1}^{\nparams}\sqrt{\lambda_k}\,\phi_k(\vv{x})\,
\params_k)$ from a truncated Karhunen--Lo\`{e}ve expansion of a
squared-exponential Gaussian process ($\nparams = 10$ retained
terms, correlation length $0.1$, standard deviation $0.3$),
supported on the left strip
$[0, 0.25] \times [0, 1]$ and zero elsewhere.
We assign an independent standard Gaussian prior to the
coefficients $\params \in \R^{\nparams}$; a representative
realization of the resulting source field is shown in the left
panel of \figref{fig:advdiff:domain}.
We then apply Bayes' rule with the likelihood of
\secref{sec:bayes_inverse} to update this prior to a posterior
over $\params$, using measurements of the tracer concentration
at the sensor locations of \figref{fig:advdiff:sensors} and the
forward model~\eqnref{eq:advdiff} to predict those measurements;
the measurements are corrupted by i.i.d.\ Gaussian noise with
standard deviation $0.1$.

\paragraph{Goal-oriented prediction.}
Bayesian inference reduces the uncertainty in the source
parameters $\params$, but the effect of this reduction on the
downstream prediction must be quantified.
We therefore push the posterior forward through a QoI map that
returns the concentration at final time $T = 4$ at the $\nqoi$
receptor locations of \figref{fig:advdiff:sensors},
$q_j(\params) = c(\vv{x}_j^{\mathrm{pred}}, T; \params)$,
placed downstream of the obstructions ($x_1 \geq 5/7$), and the
design seeks to reduce the uncertainty in this push-forward.

\paragraph{Design problem.}
Both the parameter-focused and the goal-oriented solutions depend
on the data collected, so the placement of the sensors determines
how informative the inference and prediction can be.
We therefore use the OED formulations of this paper to optimize
the measurement weights over the $\nobs$ candidate sensor
locations of \figref{fig:advdiff:sensors}.
The candidate locations are a space-filling subset of the mesh
chosen by a maximin criterion.
To determine the optimal weights $\design \in \simplex$ we
solve the design problem~\eqnref{eq:saa_problem} over the simplex
using the estimator of \secref{sec:computation}.

To solve the OED problems, the steady Navier--Stokes velocity is
computed once at the outset and reused for all prior samples, so
a per-sample forward solve involves only a transient
advection--diffusion solve driven by the sampled source field.
Each such solve is nevertheless expensive.
The outer and inner samples of \Algref{alg:estimator} are,
however, independent of $\design$
(\secref{sec:computation:algorithm}).
We therefore pre-generate a cached dataset of $(\nout + \nin)$
concentration-field samples in parallel, then reuse it for every
design optimization and cross-evaluation in the following
subsections without additional PDE solves.
With this problem setup, we can examine how risk preference shapes
the optimal design and the implications of optimizing an objective that does
not adequately reflect one's preferences.

\begin{figure}[!htbp]
  \centering
  \includegraphics[width=\textwidth]{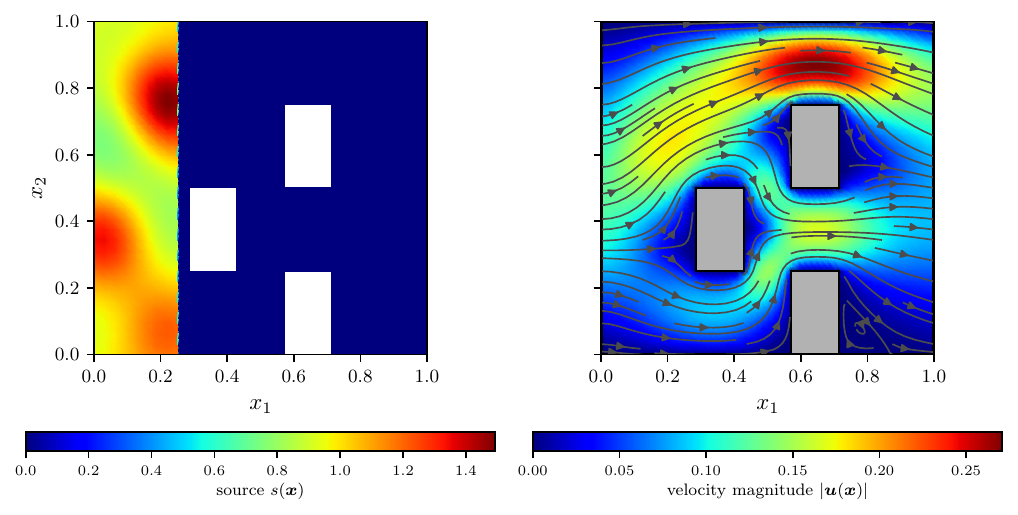}
  \caption{%
    Advection--diffusion benchmark setup on the unit square with
    three rectangular obstructions.
    \textit{Left:} a representative realization of the lognormal
    source field $s(\vv{x}; \params)$, supported on the left
    strip $[0, 0.25] \times [0, 1]$ and zero elsewhere.
    \textit{Right:} steady Navier--Stokes velocity field
    $\vv{u}(\vv{x})$ as streamlines overlaid on a contour plot
    of velocity magnitude $|\vv{u}|$, with recirculation zones
    visible around the grey obstructions.
  }
  \label{fig:advdiff:domain}
\end{figure}

\begin{figure}[!htbp]
  \centering
  \includegraphics[width=\textwidth]{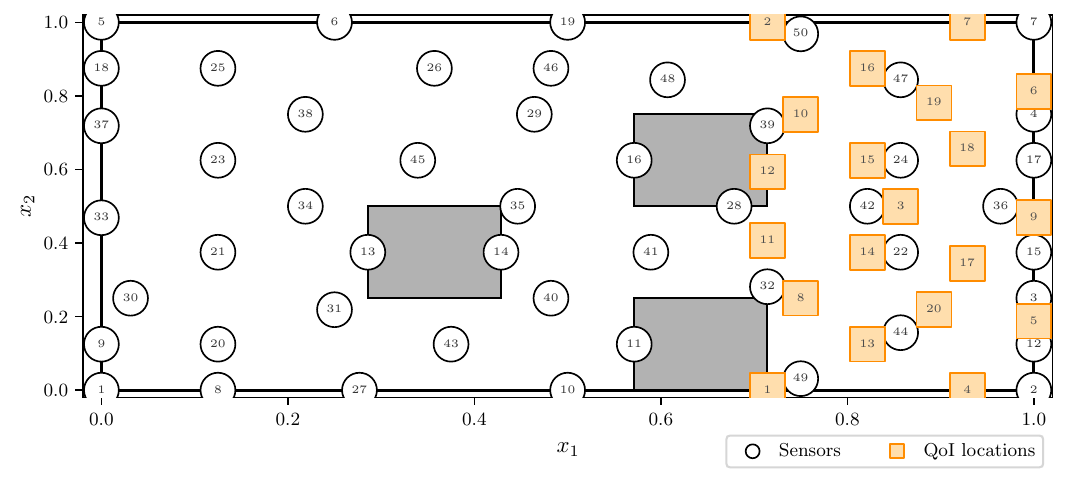}
  \caption{%
    Layout of the $\nobs = 50$ candidate sensor locations
    (white circles, numbered) and the $\nqoi = 20$ QoI
    receptor locations (orange squares, numbered) on the unit
    square; the three grey boxes are the flow obstructions.
    Sensor candidates are maximin-distributed interior mesh
    points avoiding obstruction boundaries; receptors are
    maximin-distributed within the downstream strip
    $x_1 \geq 5/7$, disjoint from sensor candidates.%
    The x-axis was stretched to ensure the sensor labels can be read. 
  }
  \label{fig:advdiff:sensors}
\end{figure}

\subsection{Risk Preference Shapes the Design}
\label{sec:experiments:advdiff:designs}

This subsection investigates the stability of optimal designs 
corresponding to various risk preferences.
Specifically, we optimize the five goal-oriented design
objectives $U_1$--$U_5$ against the parameter-focused KL
baseline given in~\Cref{tab:named_utilities}, using SciPy's
\texttt{trust-constr} algorithm with a linear simplex
constraint, convergence tolerance $\mathrm{gtol} = 10^{-8}$, and
the uniform allocation $\design = \vv{1}/\nobs$ as the starting
point.
All design objectives are estimated from the same cached dataset
with $\nout = \nin = 4000$ outer and inner samples; the AVaR
levels at the QoI and data levels are $\alpha = \beta = 0.9$,
and the entropic deviation of $U_5$ uses $\lambda = 0.9$.

To gauge how sensitive each optimized design is to the
particular simulation dataset used to construct it, we repeat the
optimization on $30$ independent bootstrap resamples of the
cached dataset (drawn with replacement at size
$(\nout, \nin)$).
The resulting $30$ designs form an empirical distribution of the
optimizer $\design^\star$ as a function of training data, which
we summarize as a sensor-wise boxplot in~\figref{fig:advdiff:boxplot}.
Here, for readability we restrict presentation to the five candidates whose mean
optimized weight exceeds a display cutoff of $0.145$ under at
least one design objective (the other $45$ candidates stay below
this cutoff under every objective and are omitted; the cutoff is
a plotting choice with no statistical meaning).
Every design objective, including the KL baseline,
concentrates almost all of its weight on at most five of the
$\nobs = 50$ candidates.
We expect the sparsity observed here is induced by a mechanism
similar to that present in asymptotic frequentist
OED~\cite{silvey1980optimal}.
Specifically, Carath\'{e}odory's theorem applied to the convex
set of positive semi-definite Fisher information matrices bounds
the support size of any optimal design measure to at most
$\tfrac{1}{2}\nparams(\nparams+1)$ atoms. 
Moreover, in practice only $\nparams$ nonzero weights suffice when the candidate set spans
$\R^{\nparams}$.

The design objectives differ in how tightly they concentrate within
this sparse support and how they split weight across the five selected locations, 
assigning the largest weights to different sensors.
This is seen in~\figref{fig:advdiff:boxplot}, where based on the median weight
across the bootstrap samples, the dominant choice 
for $U_1$, $U_2$, $U_3$, $U_4$, and $U_5$ 
Sensor~26, located above the first obstacle on the upper
branch of the flow just downstream of the source strip (see
\figref{fig:advdiff:sensors}).
However, $U_1$ instead places its largest weight on sensor~38.
The gap between the most and least weight assigned to sensor~26
by the goal-oriented design objectives exceeds $0.4$.
Thus, objective choice has first-order consequences for where
the primary measurement is sited, not just for how residual
budget is distributed.
The one exception to a clear top-sensor preference is $U_5$,
whose bootstrap box at sensor~38 ($\approx 0.34$ to $0.40$)
overlaps with $U_1$'s box at the same sensor and whose
sensor~26 weight ($\approx 0.41$) is only marginally larger.
For $U_5$ the choice between sensors~26 and~38 is genuinely
uncertain at this sample size.

Design objectives also disagree on which secondary sensors to activate.
At sensor~6, $U_1$, $U_2$, and $U_5$ place weight well above
the uniform level $1/\nobs$, while $U_3$, $U_4$, and KL fall
at or below it: some goal-oriented design objectives activate sensor~6
while others leave it effectively inactive.
A similar, smaller asymmetry is visible at sensor~34.
The activation pattern tracks tail-aversion: the strongly
tail-averse $U_3$ and $U_4$ concentrate their weight most
narrowly.
The bootstrap boxes are narrow relative to these inter-objective
gaps, so the differences are not attributable to training-set
sampling variability.

Read as a set, the five goal-oriented design objectives concentrate their
weight on sensors~26 and~38 and place only small weight
elsewhere.
The KL baseline departs from this pattern.
Its weight spreads across four sensors $[38, \hspace{0.1pt} 8, \hspace{0.1pt} 6, \hspace{0.1pt} 34]$ each receiving median weight in the range
$[0.15, 0.22]$ (\figref{fig:advdiff:boxplot}).
Among these, only sensor~$38$ overlaps with the primary
goal-oriented selection; sensors~$8$, $6$, and $34$ all sit
below the uniform line for $U_3$ and $U_4$ (the most
tail-averse goal-oriented utilities) and at or just above it
for the others.
Conversely, sensor~$26$, which receives up to $\approx 0.72$
weight from the goal-oriented utilities, receives essentially
zero from KL.
The bootstrap variability of the KL design is also noticeably
smaller than that of the goal-oriented designs, with narrower
boxes at every active sensor.
This reflects that the parameter-information criterion is
estimated more stably under resampling than the tail-sensitive
quantities that drive the risk-averse designs.

The mechanism behind these observations is that KL selects
sensors where observations constrain the source
parameters, regardless of whether those parameter directions
propagate through the advection--diffusion operator to the
downstream receptors.
\iflongappendix
The parameter-informative sensors and the prediction-informative
sensors therefore coincide only when the observation geometry
happens to align the two notions of information, and in the
advection-dominated regime they do not; \appref{app:advdiff:regimes}
shows this alignment is restored as transport becomes more
diffusive.
\fi

\subsection{The Cost of Optimizing for the Wrong Objective}
\label{sec:experiments:advdiff:cost}

The previous subsection showed that the optimal designs differ
across the risk-aware objectives.
To quantify the impact of those differences on prediction
uncertainty, we cross-evaluate the designs.
Specifically, writing $\design^\star_j$ for the design that
minimizes design objective $U_j$, we compute
$U_i(\design^\star_j)$, which is the value of objective $U_i$ at
the optimum of objective $U_j$.

Each panel of~\figref{fig:advdiff:histogram} shows the
distribution of $U_i(\design)$ over a large sample of uniformly
random simplex designs.
We also plot the vertical markers at $U_i(\design^\star_j)$ for
all six optimized designs, so the degradation incurred when
$j \neq i$ can be read off directly.
The KL-optimal design $\design^\star_{\mathrm{KL}}$ evaluated
under any of the six goal-oriented design objectives lies close to the
favorable edge of the corresponding random-design
distribution --- comparable to the best random design, but
still roughly one histogram-width away from the corresponding
goal-oriented optimum.
The reverse cross-evaluation gives the same qualitative
picture from the KL side: the six goal-oriented designs $\design^\star_j$ 
evaluated under the KL design objective perform
nearly indistinguishably from the random allocation on the simplex.
Here, optimizing for parameters yields prediction performance
no better than a random design, while optimizing for prediction
yields parameter performance no better than a random design.
These results highlight the importance of directly targeting the inference goal
when specifying an OED objective. 

In addition to highlighting the importance of prediction-oriented strategies, 
the results of~\figref{fig:advdiff:histogram} also highlight
that the way risk is encoded into a prediction-focused objective matters.
While each goal-oriented design $\design^\star_i$ evaluated under $U_j$ ($j \neq i$)
outperforms all random designs under $U_j$,
there 
are substantial differences between these objective values and the optimum $U_i(\design^\star_i)$, 
with a gap of up to $20$--$25\%$ in some cases.
It is worth noting that these differences are smaller than the
parameter-versus-goal-oriented gap, indicating that in this case, 
directly targeting the decision-making QoI is most impactful.
This is due to the fact that all goal-oriented objectives target prediction quality, 
differing 
only in how deviation is aggregated.
Consequently, they share more structure with one another than
any of them shares with the parameter-focused KL objective, and
thus a risk-aware design tuned for one risk preference retains
some of its value under another.

While misspecifying the objective class (parameters
vs.\ prediction) has the largest impact on the design, we
believe the risk preference should always be elicited and
incorporated where possible.
When it cannot be elicited easily, or agreed upon, practitioners
should instead investigate how sensitive the design is to the
risk preference.
Sensitivity can be measured by adopting the procedure in this
section, which optimizes a design for each candidate objective
and compares them.
When the cost of running the forward model is high, this
sensitivity can be computed cheaply, because all designs are
optimized and compared on the single cached sample set without
requiring additional simulations.

Together with the bootstrap sensor-weight distributions of
\figref{fig:advdiff:boxplot}, the random-design histograms of
\figref{fig:advdiff:histogram} give a practitioner two
diagnostics --- the training-data sensitivity of the optimal
design and the parameter-versus-prediction mismatch --- both
obtained from the same cached dataset and so available before
committing to an experimental design strategy, without any forward solves
beyond those used to fit the original estimator.

The diagnostics we propose highlight the relative importance of
differences between parameter- and prediction-focused
objectives, and among the risk preferences.
\iflongappendix
These differences are, however, problem-dependent;
\appref{app:advdiff:regimes} makes this concrete.
\fi
Repeating the analysis as the diffusivity $D$ is increased shows
that the effect of the risk preference weakens as transport
moves from the strongly advection-dominated regime toward a
diffusion-dominated one.
The informative region broadens, the goal-oriented and KL
designs converge, and the categorical cost of misspecifying the
objective class softens to a bounded quantitative error.
\iflongappendix
The appendix also shows how to convert a design computed using
any objective in any physics regime into an integer measurement
budget, using a largest-remainder rule that preserves the weight
hierarchy of the continuous design.
\fi
This closes the gap between the continuous optima analyzed above
and the discrete allocation a practitioner ultimately deploys.

\begin{figure}[!htbp]
  \centering
  \includegraphics[width=\textwidth]{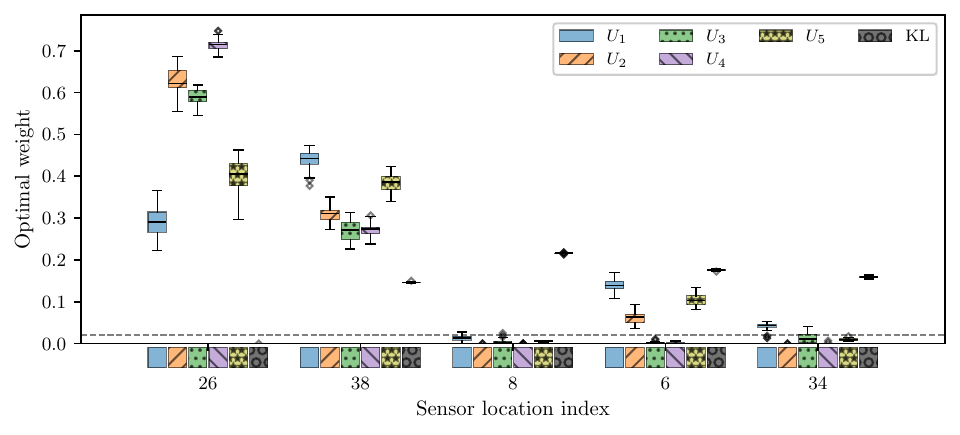}
  \caption{%
    Bootstrap distribution of optimal sensor weights in the
    advection-dominated regime, across $30$ resampled training
    sets.
    Shown are the five sensor locations whose mean weight exceeds
    a display cutoff of $0.145$ under at least one design
    objective ($45$ other candidates omitted for readability).
    Box plots grouped by sensor index (x-axis) and colored by
    design objective, with the uniform allocation $1/\nobs$ marked by
    the dashed horizontal line.%
  }
  \label{fig:advdiff:boxplot}
\end{figure}
\begin{figure}[!htbp]
  \centering
  \includegraphics[width=\textwidth]{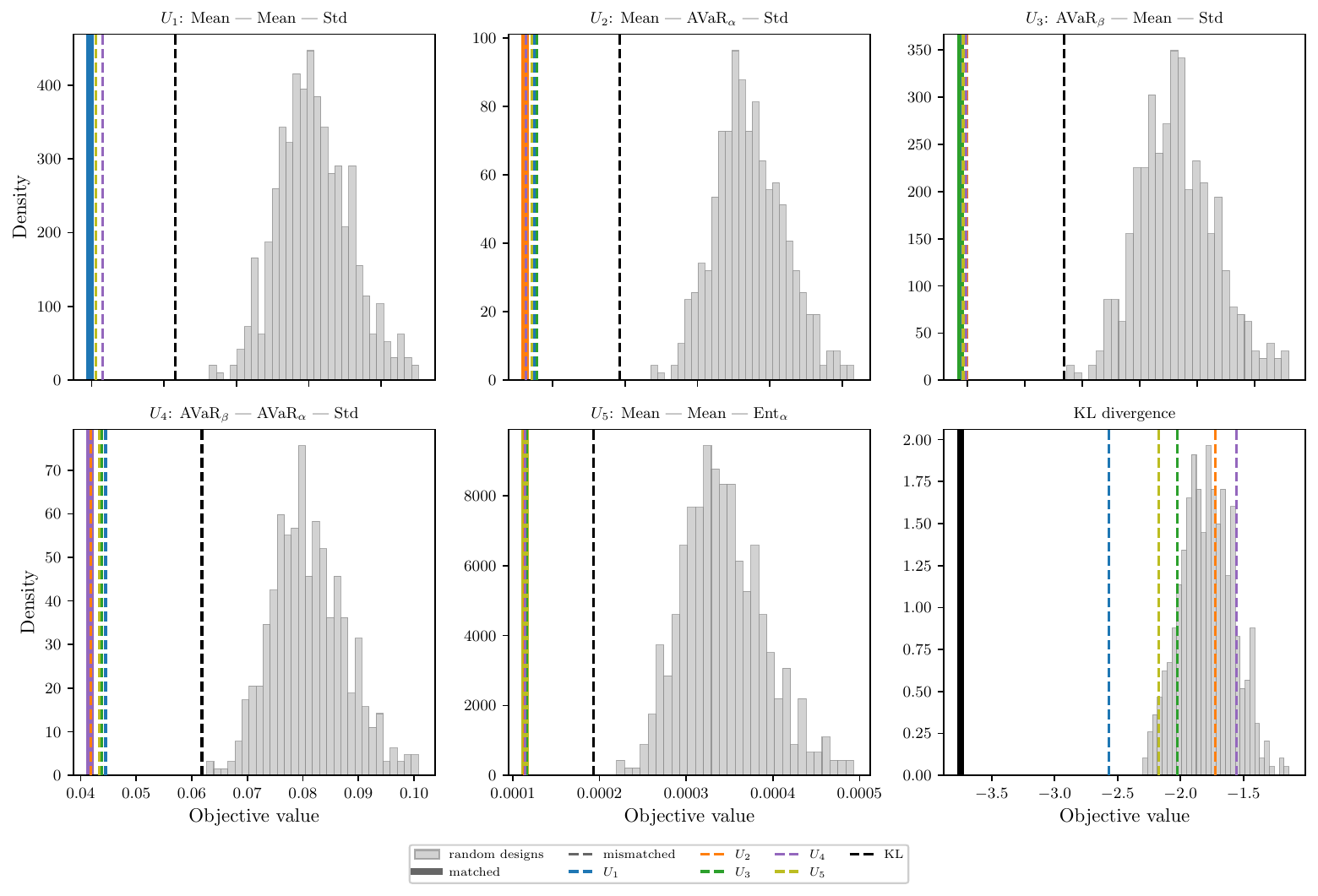}
  \caption{%
    Design-objective value histograms over randomly generated simplex
    designs, one panel per design objective, advection-dominated regime.
    Vertical markers: objective value at the six optimized
    designs $\design^\star_j$ evaluated under that panel's
    design objective, where lower is more optimal.
  }
  \label{fig:advdiff:histogram}
\end{figure}

\section{Conclusion}
\label{sec:conclusion}

This paper presents a novel framework for incorporating risk preferences
into goal-oriented Bayesian optimal experimental design paradigms.
Standard approaches aggregate predictive uncertainty
using expectations, deviations, or divergences that
implicitly encode, often risk neutral, preferences.
This can lead to suboptimal data collection, especially in goal-oriented contexts
where uncertainties regarding rare or high-consequence events may not be captured 
by optimizing for average behavior.
Our risk-aware framework leverages the risk quadrangle to construct 
a three-level OED objective separating posterior deviation, 
aggregation of vector-valued prediction quantities of interest (QoIs), and 
aggregation over the likely data.
We provide a corresponding algorithm for numerical implementation that 
builds off existing double-loop estimators 
and show how common deviation and risk measure choices---such as the standard deviation, average value-at-risk (AVaR), 
entropic risk, and safety-margin design--- fit within this framework and implementation strategy.

We demonstrate this approach on a closed-form illustrative problem as well as an advection-diffusion PDE example, 
where the observable is nonlinear in the inference parameters.
These test cases, illustrate how parameter-oriented approaches can underperform when the ultimate goal is predicting a downstream QoI, motivating the use of goal-oriented strategies.
Furthermore, we compare six goal-oriented designs reflecting differing risk preferences, 
and show that such preferences can substantially impact the resulting design:
risk-averse and risk-neutral
design objectives place their primary measurement at physically distinct
locations.
We also evaluate the sensitivity of the optimal design strategy to the choice of risk preference 
through cross-comparison of design performance in each of the six goal-oriented objectives along with 
the parameter-oriented objective.

While effective, the approach presented has two primary limitations, which result from encoding 
the design as a diagonal noise-weighting matrix 
for convenience in continuous-simplex
optimization.
The first, is that this choice relies on the
observation model being Gaussian.
Extension to non-Gaussian likelihoods requires a different
encoding of $\design$ in the likelihood, or a different posterior
estimation scheme.
The second issue is that when sensor noise is extremely large, 
this choice of parameterization results in poor scaling of the nested importance-weighted estimators beyond approximately $100$ candidate sensor locations.

Beyond addressing the two limitations above, future work should
focus on extending this work to sequential OED and on reducing
its computational cost by replacing the importance-sampling step
with amortized generative models.
Sequential strategies that update the posterior iteratively 
would fit naturally into the algorithmic framework presented here, 
as importance weights can be updated recursively.
However, both sequential and batched OED (presented here) could benefit from using 
generative models, such as those leveraging normalizing flows, to reduce cost.
Taken together, these extensions would broaden the reach of
risk-aware goal-oriented design beyond the regime the present
estimator supports to 
larger candidate sets, non-Gaussian likelihoods, and sequential
settings encountered in practice, while preserving the principled
treatment of risk preference, the importance of which is established in this work.

\section*{Acknowledgments}
This paper describes objective technical results and analysis. Any subjective
views or opinions that might be expressed in the paper do not necessarily represent
the views of the U.S. Department of Energy or the United States Government. Sandia
National Laboratories is a multimission laboratory managed and operated by National
Technology and Engineering Solutions of Sandia, LLC., a wholly owned subsidiary of
Honeywell International, Inc., for the U.S. Department of Energy’s National Nuclear
Security Administration under contract DE-NA-0003525.
This material is based upon work supported by the U.S. Department of Energy,
Office of Science, Office of Advanced Scientific Computing Research (ASCR), under the contract
24-028431, the Field Work Proposal Number 23-02526, and Early Career Research Program as well as Sandia National Laboratories' Computing and Information Sciences Laboratory Directed Research and Development program.

\section*{Declaration of Generative AI in the Writing Process}
During the preparation of this work the authors used SandiaAI - Chat and Claude (Anthropic) to assist with editing and phrasing, and verify the manuscript's consistency with the computational results. After using this tool/service, the authors reviewed and edited the content as needed and take full responsibility for the content of the publication.


\bibliography{references}


\appendix

\iflongappendix

\section{Linear Gaussian Model: Posterior and Push-Forward Statistics}
\label{app:model}

We collect the posterior and push-forward statistics for the
linear Gaussian model used in the running example and experiments.
All subsequent appendices use the notation defined here.

\subsection*{A.1\quad Model specification}

Let $\params \in \R^{\nparams}$ be the parameter vector with
prior $\params \sim \Gaussian{\priormean}{\priorcov}$.
Observations follow the linear model
\begin{equation}
  \label{eq:model:obs}
  \data = \obsmtx\,\params + \noisevec,
  \qquad
  \noisevec \mid \design \sim \Gaussian{\vv{0}}{\effnoisecov},
\end{equation}
where $\obsmtx \in \R^{\nobs \times \nparams}$ is the observation
(Vandermonde) matrix with entries $A_{ij} = x_i^{\,j-1}$ evaluated
at candidate locations $x_1,\ldots,x_{\nobs}$,
and the baseline noise is diagonal,
$\noisecov = \diag(\sigma_{\varepsilon,1}^2,\ldots,\sigma_{\varepsilon,\nobs}^2)$,
with independent, possibly heteroscedastic noise across candidates.
Allocating weight $\designk{k}$ to candidate $k$ scales its precision,
giving the diagonal effective noise covariance
$\effnoisecov = \diag(\sigma_{\varepsilon,k}^2/\designk{k})$,
so a larger weight $\designk{k}$ reduces the effective noise at
location $k$.
Its inverse, $\inv{\effnoisecov} = \diag(\designk{k}/\sigma_{\varepsilon,k}^2)$,
appears throughout the posterior expressions below.
In the polynomial-basis experiments we use a degree-4 basis ($\nparams = 5$),
$\nobs = 10$ candidate locations equally spaced in $[-1,1]$,
a common noise standard deviation $\sigma_{\varepsilon,k} = 0.5$ for all $k$, and
isotropic prior $\priorcov = \sigma_0^2 \mm{I}$ with $\sigma_0 = 0.5$.

\subsection*{A.2\quad Posterior distribution}

Because the model is linear-Gaussian, the posterior is Gaussian:
$\params \mid \data, \design \sim \Gaussian{\postmeanb}{\postcovb}$
with
\begin{align}
  \postcovb  &= \bigl(\inv{\priorcov} + \obsmtx\T\inv{\effnoisecov}\obsmtx
                \bigr)^{-1},
  \label{eq:post:cov}\\
  \postmeanb &= \postcovb\bigl(\obsmtx\T\inv{\effnoisecov}\data
                + \inv{\priorcov}\priormean\bigr).
  \label{eq:post:mean}
\end{align}
Crucially, $\postcovb$ depends only on the design $\design$,
not on the data $\data$.

\subsection*{A.3\quad Distribution of the posterior mean over data}

Since $\data = \obsmtx\params + \noisevec$ with
$(\params,\noisevec)$ independent, $\postmeanb$ is a Gaussian
random variable under the marginal data distribution
$\data \sim p(\data \mid \design)$.
Define
\begin{equation}
  \label{eq:R-proc}
  \mm{R} \coloneqq \postcovb\,\obsmtx\T\inv{\effnoisecov},
  \qquad
  \mm{R}_0 \coloneqq \mm{R}\,\obsmtx.
\end{equation}
Because $\postmeanb = \mm{R}\,\data + \postcovb\inv{\priorcov}\priormean$
and $\Eunder{p(\data\mid\design)}{\data} = \obsmtx\priormean$, the mean of $\postmeanb$ over
data is
\begin{align}
  \pmnu &\coloneqq \Eunder{p(\data\mid\design)}{\postmeanb}
         = \mm{R}\,\obsmtx\priormean
           + \postcovb\inv{\priorcov}\priormean
         = (\mm{R}_0 + \postcovb\inv{\priorcov})\priormean \notag\\
        &= \postcovb\bigl(\obsmtx\T\inv{\effnoisecov}\obsmtx
           + \inv{\priorcov}\bigr)\priormean
         = \postcovb\,\inv{\postcovb}\priormean
         = \priormean.
  \label{eq:pmnu}
\end{align}
Hence $\pmnu = \priormean$: the posterior mean is unbiased for
the prior mean in expectation.
The covariance of $\postmeanb$ over data is
\begin{equation}
  \label{eq:pmcov}
  \pmcov \coloneqq \Cov{\postmeanb}{\postmeanb}
  = \mm{R}\,\Cov{\data}{\data}\,\mm{R}\T
  = \mm{R}_0\,\priorcov\,\mm{R}_0\T
    + \mm{R}\,\noisecov\,\mm{R}\T,
\end{equation}
where $\Cov{\data}{\data} = \obsmtx\priorcov\obsmtx\T + \noisecov$.

\subsection*{A.4\quad Lognormal push-forward}

For the QoI map $\qoicomp{j} = \exp(\predvec_j\T\params)$,
define the log-mean
$\lntauj{j} \coloneqq \predvec_j\T\postmeanb$
and the posterior log-variance
$\lnsigmasqj{j} \coloneqq \predvec_j\T\postcovb\,\predvec_j$.
Since $\predvec_j\T\params \mid \data,\design \sim
\Gaussian{\lntauj{j}}{\lnsigmasqj{j}}$,
each QoI component follows a lognormal distribution:
$\qoicomp{j} \mid \data,\design
\sim \LogNormal{\lntauj{j}}{\lnsigmasqj{j}}$.
Its conditional standard deviation factorizes as
\begin{equation}
  \label{eq:Dj:model}
  d_j(\data,\design) \equiv \Std{\qoicomp{j}\mid\data,\design}
  = \lnscalej{j}\,\exp(\lntauj{j}),
\end{equation}
where the \emph{deterministic scale factor}
\begin{equation}
  \label{eq:Kj:model}
  \lnscalej{j} = \exp\!\bigl(\tfrac{1}{2}\lnsigmasqj{j}\bigr)
             \sqrt{\exp(\lnsigmasqj{j}) - 1}
\end{equation}
depends on $\design$ only through $\postcovb$.
The log-mean $\lntauj{j}$ is Gaussian over data with
\begin{equation}
  \label{eq:tauj:dist}
  \lnnuj{j} \coloneqq \E{\lntauj{j}} = \predvec_j\T\priormean,
  \qquad
  \lnsigmatauj{j} \coloneqq \Var{\lntauj{j}}
                    = \predvec_j\T\pmcov\,\predvec_j.
\end{equation}


\section{Analytical Design-Objective Expressions: Scalar \texorpdfstring{QoI}{QoI}}
\label{app:scalar}

We derive closed-form expressions for the design objectives
$U_1$--$U_4$ and the EIG baseline for the scalar-QoI case
($\nqoi = 1$).
Throughout, $\lnscale$, $\lntau$, $\lnnu$, $\lnsigmasq$,
and $\lnsigmatau$ are as defined in \appref{app:model}
(dropping the component subscript $j$ since $\nqoi = 1$).
In particular $d(\data,\design) = \lnscale\exp(\lntau)$
with $\lntau \sim \Gaussian{\lnnu}{\lnsigmatau}$
marginally over data.

\subsection{$U_1$: Expected standard deviation}

\begin{proposition}
\label{prop:u1}
$U_1(\design) = \Eunder{p(\data\mid\design)}{d(\data,\design)}
= \lnscale\exp\!\bigl(\lnnu + \tfrac{1}{2}\lnsigmatau\bigr).$
\end{proposition}
\begin{proof}
Since $d = \lnscale e^{\lntau}$ and
$\lntau \sim \Gaussian{\lnnu}{\lnsigmatau}$,
the random variable $d/\lnscale$ is lognormal.
The mean of a lognormal $\LogNormal{\mu}{\sigma^2}$ is
$e^{\mu + \sigma^2/2}$, giving the result.
\end{proof}

\subsection{$U_3$: \texorpdfstring{AVaR$_\beta$}{AVaR} of the standard deviation over data}

\begin{proposition}
\label{prop:u3}
\begin{equation}
  \label{eq:U3}
  U_3(\design)
  = \AVaR{\beta}{d(\data,\design)}
  = \frac{\lnscale\exp\!\bigl(\lnnu + \tfrac{1}{2}\lnsigmatau\bigr)}
         {1-\beta}\;
    \normcdf{\sqrt{\lnsigmatau} - \normppf{\beta}}.
\end{equation}
\end{proposition}
\begin{proof}
The random variable $d = \lnscale e^\lntau$ is lognormal with
parameters $\log\lnscale + \lnnu$ and $\lnsigmatau$.
For a lognormal $X \sim \LogNormal{\mu_X}{\sigma_X^2}$, the
$\beta$-quantile is $q_\beta = e^{\mu_X + \sigma_X\normppf{\beta}}$
and the conditional tail expectation is
\[
  \AVaR{\beta}{X}
  = \frac{\Eunder{p(\data\mid\design)}{X}}{1-\beta}\,
    \normcdf{\sigma_X - \normppf{\beta}}.
\]
Substituting $\mu_X = \log\lnscale + \lnnu$,
$\sigma_X = \sqrt{\lnsigmatau}$, and
$\Eunder{p(\data\mid\design)}{X} = \lnscale\exp(\lnnu + \lnsigmatau/2)$
gives \eqnref{eq:U3}.
\end{proof}

\subsection{$U_4$: Safety margin}

\begin{proposition}
\label{prop:u4}
For safety factor $c \ge 0$,
\begin{equation}
  \label{eq:U4}
  U_4(\design)
  = \Eunder{p(\data\mid\design)}{d} + c\,\Std{d(\data,\design)}
  = \lnscale\exp\!\bigl(\lnnu + \tfrac{1}{2}\lnsigmatau\bigr)
    \Bigl(1 + c\sqrt{\exp(\lnsigmatau) - 1}\Bigr).
\end{equation}
\end{proposition}
\begin{proof}
$\Eunder{p(\data\mid\design)}{d} = U_1$ from Proposition~\ref{prop:u1}.
For the lognormal $d/\lnscale \sim \LogNormal{\lnnu}{\lnsigmatau}$:
$\Varunder{p(\data\mid\design)}{d} = \lnscale^2 e^{2\lnnu+\lnsigmatau}(e^{\lnsigmatau}-1)$,
so $\Stdunder{p(\data\mid\design)}{d} = \lnscale e^{\lnnu+\lnsigmatau/2}\sqrt{e^{\lnsigmatau}-1}
= U_1\sqrt{e^{\lnsigmatau}-1}$.
Substituting into $U_4 = \Eunder{p(\data\mid\design)}{d} + c\,\Stdunder{p(\data\mid\design)}{d}$ gives \eqnref{eq:U4}.
\end{proof}

\subsection{EIG baseline}

The KL-based baseline used in the numerical experiments is the
expected information gain (EIG) on the \emph{full} parameter
posterior, not its projection onto the prediction direction
$\predvec$.
For the linear-Gaussian model of \appref{app:model} this EIG admits
a closed form that is data-independent because the posterior
covariance does not depend on the realized data.

\begin{proposition}
\label{prop:kl}
Let $\fisherb(\design) = \obsmtx\T \inv{\effnoisecov}\, \obsmtx$
denote the design-weighted Fisher information.
Then
\begin{equation}
  \label{eq:KL}
  U_{\mathrm{KL}}(\design)
  \coloneqq \Eunder{p(\data\mid\design)}{
      \KL{p(\params \mid \data, \design)}{p(\params)}
    }
  = \tfrac{1}{2}\log\frac{\det\priorcov}{\det\postcovb}
  = \tfrac{1}{2}\log\det\!\left(
      \mm{I} + \priorcov\,\fisherb(\design)
    \right).
\end{equation}
\end{proposition}
\begin{proof}
For Gaussian
$p(\params) = \Gaussian{\priormean}{\priorcov}$ and
$p(\params \mid \data, \design) = \Gaussian{\postmeanb}{\postcovb}$,
the KL divergence has the standard form
\begin{equation}
  \label{eq:KL:gauss}
  \KL{p(\params \mid \data, \design)}{p(\params)}
  = \tfrac{1}{2}\!\left[
      \tr\!\left(\inv{\priorcov}\,\postcovb\right)
      - \nparams
      + \log\frac{\det\priorcov}{\det\postcovb}
      + (\postmeanb - \priormean)\T \inv{\priorcov}\,
        (\postmeanb - \priormean)
    \right].
\end{equation}
Taking the expectation over $\data \sim p(\data \mid \design)$:
the first three terms are $\data$-independent
(\appref{app:model}~A.2 gives $\postcovb$ independent of $\data$);
for the quadratic term,
$\Eunder{p(\data\mid\design)}{\postmeanb} = \priormean$
\eqnref{eq:pmnu} and
$\Cov{\postmeanb}{\postmeanb} = \pmcov$
\eqnref{eq:pmcov} give
\[
  \Eunder{p(\data\mid\design)}{
      (\postmeanb - \priormean)\T \inv{\priorcov}
      (\postmeanb - \priormean)
    }
  = \tr\!\left(\inv{\priorcov}\,\pmcov\right).
\]
The law of total variance, applied to $\params$ under
$p(\params,\data \mid \design)$ and using that $\postcovb$ is
$\data$-independent, gives
$\priorcov = \postcovb + \pmcov$, so
$\tr(\inv{\priorcov}\,\pmcov)
 = \nparams - \tr(\inv{\priorcov}\,\postcovb)$,
which cancels the first two terms of~\eqnref{eq:KL:gauss}.
The surviving log-determinant gives the middle expression
in~\eqnref{eq:KL}.
The Fisher form follows from
$\inv{\postcovb} = \inv{\priorcov} + \fisherb(\design)$
\eqnref{eq:post:cov}, which gives
$\priorcov\,\inv{\postcovb} = \mm{I} + \priorcov\,\fisherb(\design)$
and hence
$\det(\priorcov)/\det(\postcovb) = \det(\mm{I} + \priorcov\,\fisherb(\design))$.
\end{proof}


\section{Convergence of the Design-Objective Estimator}
\label{app:convergence}

\iflongappendix
This appendix verifies that the nested-quadrature design-objective
estimator of \secref{sec:computation:algorithm} converges to the
closed-form references of \appref{app:scalar}, quantifies how its
accuracy depends on the number of inner and outer quadrature
points, and compares Monte Carlo (MC) sampling against
quasi-Monte Carlo (QMC) using randomly-shifted Halton sequences.
\else
This appendix verifies that the nested-quadrature design-objective
estimator of \secref{sec:computation:algorithm} converges to
closed-form reference values, derived analytically from the
Gaussian conjugacy of the posterior and the moments and quantiles
of the lognormal push-forward of the QoI, quantifies how its
accuracy depends on the number of inner and outer quadrature
points, and compares Monte Carlo (MC) sampling against
quasi-Monte Carlo (QMC) using randomly-shifted Halton sequences.
\fi

\iflongappendix
To make the comparison against an exact reference possible, we
use the linear Gaussian model of
\secref{sec:linear_gaussian_running_example}, the one setting in
which the design objectives admit the closed forms of
\appref{app:scalar}.
\else
To make the comparison against an exact reference possible, we
use the linear Gaussian model of
\secref{sec:linear_gaussian_running_example}, the one setting in
which the design objectives admit closed forms, obtained from
standard Gaussian posterior algebra and the analytical mean,
standard deviation, and quantiles of the resulting lognormal
push-forward.
\fi
We take a degree-$p = 4$ polynomial basis
$\psi(x) = [1,\, x,\, x^2,\, x^3,\, x^4]\T$, so that
$\obsmtx \in \R^{\nobs \times (p+1)}$ is a Vandermonde matrix
evaluated at the $\nobs = 10$ equidistant observation locations on
$[-1, 1]$.
The QoI map~\eqnref{eq:lognormal_qoi} is evaluated at a single
prediction location $x_{\mathrm{pred}} = 0$, giving a scalar QoI
($\nqoi = 1$), and we use the prior
$\params \sim \Gaussian{\vv{0}}{0.25\,\mm{I}}$ with base noise
covariance $\noisecov = 0.25\,\mm{I}$.

We evaluate the estimator at the uniform design
$\design = \vv{1}/\nobs$ for the four named design objectives whose
analytical values are available --- $U_1$ (std / mean / mean),
$U_3$ (std / mean / $\mathrm{AVaR}_\beta$, $\beta = 0.9$), $U_6$
(std / mean / mean$+c\,\cdot\,$Std, $c = 1$), and KL.
For each, we generate $R = 1000$ independent data realizations,
compute the estimate $\hat{U}_{\nout,\nin}^{(r)}$ on each, and
form the Monte Carlo approximation to the mean squared error
\begin{equation}
  \label{eq:mse}
  \mathrm{MSE}(\hat{U}_{\nout,\nin})
  \approx \frac{1}{R}\sum_{r=1}^{R}
    \bigl(U - \hat{U}_{\nout,\nin}^{(r)}\bigr)^2
  = \mathrm{Bias}^2(\hat{U}_{\nout,\nin})
  + \mathrm{Var}(\hat{U}_{\nout,\nin}),
\end{equation}
where $U$ is the analytical value.
Estimating the squared bias
$(\E{\hat{U}_{\nout,\nin}} - U)^2$ and the variance
$\E{(\hat{U}_{\nout,\nin} - \E{\hat{U}_{\nout,\nin}})^2}$ from the
same $R$ replicates separates the two sources of error, so we can
see which one limits accuracy in each regime.
For QMC the replicates correspond to independent Halton shifts,
which preserves this bias--variance decomposition.

\figref{fig:exp1:convergence} reports the two error components as
the inner and outer sample counts are swept in turn ($\nin$
varied at $\nout = 4000$, then $\nout$ varied at $\nin = 4000$).
For the prediction objectives $U_1$, $U_3$, and $U_6$, refining
the inner loop drives both components down, and because the
variance stays above the squared bias throughout, it is the
variance that sets the achievable accuracy.
The KL objective behaves differently, and for a structural
reason: its log-sum-exp form converges rapidly in the inner loop,
so once $\nin \gtrsim 200$ the inner contribution is already
negligible and the total error plateaus at a floor set by the
outer loop.
Adding inner samples beyond this point reduces the Jensen bias of
the KL estimator but cannot touch the outer-loop sampling
variance that now dominates, so the MSE stops improving.

The outer-loop sweep confirms this division of labor.
For the prediction objectives both bias and variance again fall,
with variance dominating, so increasing $\nout$ continues to buy
accuracy.
KL moves in the opposite sense from its inner-loop behavior: its
variance decreases monotonically with $\nout$ while its squared
bias stays essentially constant, since the bias is fixed by
$\nin$ and is therefore independent of how many outer samples are
drawn.
Across every objective, QMC outer sampling reaches lower variance
than MC at matched sample counts, a direct consequence of the
low-discrepancy structure of the Halton sequence.
Together these sweeps show the estimator converging to the
analytical references and identify which loop to refine for a
given objective: the inner loop for the prediction objectives,
the outer loop for KL.

\begin{figure}[htb]
  \centering
  \includegraphics[width=\textwidth]{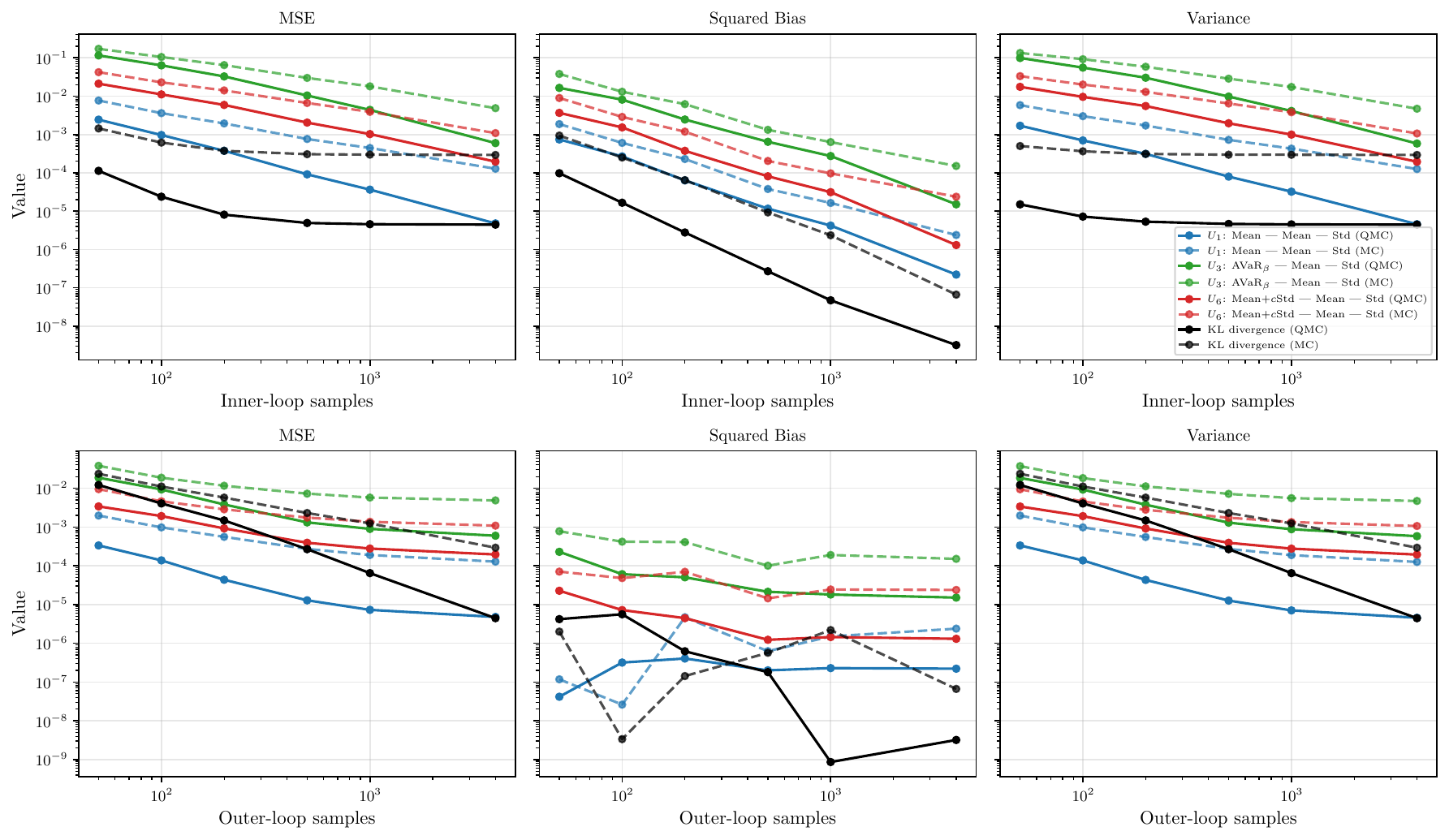}
  \caption{%
    MSE decomposition ($\mathrm{MSE} = \mathrm{Bias}^2 + \mathrm{Var}$)
    for MC (dashed) and QMC (solid) estimators of the
    $U_1$, $U_3$, $U_6$, and KL utilities, evaluated at uniform
    design weights.
    \textit{Top row:} inner-loop sweep ($\nin$ varied,
    $\nout = 4000$ fixed).
    \textit{Bottom row:} outer-loop sweep ($\nout$ varied,
    $\nin = 4000$ fixed).
    The decoupled estimator of~\secref{sec:computation:algorithm}
    shares a single inner sample set across all outer samples.%
  }
  \label{fig:exp1:convergence}
\end{figure}


\section{Gradient Derivations}
\label{app:gradients}

We derive exact gradients of the sample-average OED objective with
respect to the design weights $\design$, using fixed quadrature
samples drawn once before optimization.
Section~\ref{app:gradients:general} presents the general three-level
chain rule; Section~\ref{app:gradients:u1} specializes to the
Std/Mean/Mean configuration (design objective~U1).

\subsection{Reparameterized observations and log-likelihood}
\label{app:gradients:general}

Fix outer samples
$\{(\outersamp{m}, \noisem{m})\}_{m=1}^{\nout}$
from the joint prior--noise distribution, with
$\noisem{m} \sim \Gaussian{\vv{0}}{\mm{I}}$, and inner prior samples
$\{\innersamp{n}\}_{n=1}^{\nin}$.
As in \appref{app:model}, the baseline noise covariance is diagonal,
$\noisecov = \diag(\sigma_{\varepsilon,1}^2,\ldots,\sigma_{\varepsilon,\nobs}^2)$,
and the design enters as the diagonal weighting
$\weightmtx = \diag(\design)$, so the reparameterization that
generates the artificial observations separates across sensors.
For each outer sample $m$ and sensor $k$,
\begin{equation}
  \label{eq:app:reparam}
  y_k^{(m)}
  = \obsmtx_k\T\outersamp{m}
    + \frac{\sigma_{\varepsilon,k}}{\sqrt{\designk{k}}}\,z_k^{(m)},
\end{equation}
where $z_k^{(m)}$ is the $k$-th entry of $\noisem{m}$ and the
latent noise realizations are drawn once and held fixed throughout
optimization.
This makes each $\datam{m}$ a deterministic, differentiable
function of $\design$, and because the weighting acts
sensor-by-sensor, reducing $\designk{k}$ toward zero inflates the
effective noise on sensor $k$ toward infinity.

The log-likelihood inherits the same separation across sensors.
For inner sample $n$ and outer realization $m$, the log-likelihood is
\begin{equation}
  \label{eq:loglike}
  \ell_{mn}(\design)
  = -\tfrac{1}{2}\sum_{k=1}^{\nobs}
    \left[\log\frac{\sigma_{\varepsilon,k}^2}{\designk{k}}
    + \frac{\designk{k}\,\residk{m}{n}{k}^2}{\sigma_{\varepsilon,k}^2}\right]
    + \mathrm{const},
\end{equation}
where $\residk{m}{n}{k} = y_k^{(m)} - \obsmtx_k\T\innersamp{n}$ is the
residual at location $k$.
Because $y_k^{(m)}$ depends on $\design$ through
\eqnref{eq:app:reparam}, the partial derivative has three contributions:
\begin{equation}
  \label{eq:loglike-jac}
  \frac{\partial\ell_{mn}}{\partial\designk{k}}
  = \underbrace{\frac{1}{2\designk{k}}}_{\text{log-det}}
    \underbrace{-\;\frac{\residk{m}{n}{k}^2}{2\sigma_{\varepsilon,k}^2}}_{\text{quadratic}}
    \underbrace{+\;\frac{\residk{m}{n}{k}\,z_k^{(m)}}
                        {2\sigma_{\varepsilon,k}\sqrt{\designk{k}}}}_{\text{reparameterization}}.
\end{equation}
The first two terms arise from the Gaussian log-density; the third
accounts for the implicit design-dependence of $\datam{m}$.

\subsection{Evidence and normalized importance weights}

The marginal likelihood (evidence) for outer sample $m$ is
\begin{equation}
  \label{eq:app:evidence}
  \hat p(\datam{m} \mid \design)
  = \sum_{n=1}^{\nin} \inwt{n}\,
    \exp\!\bigl(\ell_{mn}(\design)\bigr),
\end{equation}
where $\inwt{n}$ are inner quadrature weights (uniform
$1/\nin$ for Monte Carlo).
The normalized importance weights
\begin{equation}
  \label{eq:app:impweights}
  \impwt{m}{n}
  = \frac{\inwt{n}\,\exp\!\bigl(\ell_{mn}(\design)\bigr)}
         {\hat p(\datam{m} \mid \design)}
\end{equation}
approximate the posterior $p(\innersamp{n}\mid\datam{m},\design)$
and satisfy $\sum_n \impwt{m}{n} = 1$ for each $m$.
Their gradient with respect to $\designk{k}$ follows from the
quotient rule:
\begin{equation}
  \label{eq:impweights-jac}
  \frac{\partial\impwt{m}{n}}{\partial\designk{k}}
  = \impwt{m}{n}\!\left[
      \frac{\partial\ell_{mn}}{\partial\designk{k}}
      - \sum_{n'} \impwt{m}{n'}\,
        \frac{\partial\ell_{mn'}}{\partial\designk{k}}
    \right],
\end{equation}
where the bracketed term is the importance-weighted mean of the
log-likelihood gradient, required to propagate the gradient
through the normalization.

\subsection{Chain rule through the three levels}

The sample-average approximation of \eqnref{eq:utility:full} is
\begin{equation}
  \hat{U}(\design)
  = \datarisk\!\left[
      \qoirisk\!\left[
        \hat{d}_j(\datam{m}, \design)
      \right]_{j=1}^{\nqoi}
    \right]_{m=1}^{\nout},
\end{equation}
where $\hat{d}_j$ is the deviation estimate at prediction
component $j$ and outer sample $m$.
The gradient is assembled by applying the chain rule through each
level in turn.

\begin{enumerate}

\item \textbf{Level~1 (deviation):}
For each prediction component $j$ and outer sample $m$, compute
$\hat{d}_{jm}(\design)$ and its partial derivatives
$\partial\hat{d}_{jm}/\partial\designk{k}$
(see Sections~\ref{app:gradients:u1} and~\ref{app:smooth-avar}
for specific deviation types).

\item \textbf{Level~2 (QoI risk):}
Apply $\qoirisk$ to
$(\hat{d}_{1m},\ldots,\hat{d}_{\nqoi m})$ for each $m$ to obtain
the per-realization scalar $\hat{r}_m(\design)$.
The chain rule gives
\[
  \frac{\partial\hat{r}_m}{\partial\designk{k}}
  = \sum_{j=1}^{\nqoi}
    \frac{\partial\qoirisk}{\partial\hat{d}_{jm}}
    \cdot
    \frac{\partial\hat{d}_{jm}}{\partial\designk{k}}.
\]
For the mean, $\partial\qoirisk/\partial\hat{d}_{jm} = 1/\nqoi$.
For $\mathrm{AVaR}_\alpha$, the partial is the projection weight
$\pi^*_{jm}$ from \S~D.4 below.
For the entropic risk
$\frac{1}{\lambda}\log\frac{1}{\nqoi}\sum_j e^{\lambda\hat{d}_{jm}}$,
it is the exponentially tilted weight
$e^{\lambda\hat{d}_{jm}}/\sum_{j'}e^{\lambda\hat{d}_{j'm}}$.

\item \textbf{Level~3 (data risk):}
Apply $\datarisk$ to $(\hat{r}_1,\ldots,\hat{r}_{\nout})$
to obtain the objective estimate $\hat{U}(\design)$.
The gradient is
\[
  \frac{\partial\hat{U}}{\partial\designk{k}}
  = \sum_{m=1}^{\nout}
    \frac{\partial\datarisk}{\partial\hat{r}_m}
    \cdot
    \frac{\partial\hat{r}_m}{\partial\designk{k}},
\]
with the same risk-measure-specific weights as in Level~2.

\end{enumerate}

The complete Jacobian $\nabla_{\design} \hat{U}(\design)$ requires
only quantities already computed in the forward pass, making the
gradient cost a constant multiple of the function evaluation cost.

\subsection{Smoothed AVaR evaluation and gradient}
\label{app:smooth-avar}

The standard AVaR at level $\alpha \in [0,1)$ of a discrete
distribution with samples $\{f_i\}_{i=1}^{N}$ and probability
weights $\{v_i\}_{i=1}^{N}$ summing to one is
\begin{equation}
  \label{eq:avar:def}
  \AVaR{\alpha}{f}
  = \VaR{\alpha}{f}
    + \frac{1}{1-\alpha}\,
      \sum_{i=1}^{N} v_i\,\max\!\left(0,\, f_i - \VaR{\alpha}{f}\right),
\end{equation}
where $\VaR{\alpha}{f}$ is the $\alpha$-quantile.
The positive-part function $t\mapsto\max(0,t)$ makes
\eqnref{eq:avar:def} non-differentiable.
We use a regularization of the dual
(risk-envelope) representation~\cite{kouri2020epi}.
We note that one could instead leverage the primal-dual risk minimization algorithm \cite{kouri2022primal} to avoid smoothing.

\paragraph{Dual representation and smoothing.}
The exact AVaR admits the dual representation
\begin{equation}
  \AVaR{\alpha}{f}
  = \max_{\vv{\pi}\in\mathcal{C}_\alpha}
    \sum_{i=1}^{N} \pi_i f_i,
\end{equation}
where the AVaR risk envelope is
$\mathcal{C}_\alpha = \bigl\{
  \vv{\pi}\in\R^N : 0 \le \pi_i/v_i \le (1-\alpha)^{-1},\;
  \sum_i \pi_i = 1
\bigr\}$.
The smoothed version is the Moreau-Yosida approximation:
\begin{equation}
  \label{eq:smooth-avar-E4}
  \widetilde{\mathrm{AVaR}}_{\alpha,\delta}(f)
  = \max_{\vv{\pi}\in\mathcal{C}_\alpha}
    \left\{
      \sum_{i=1}^{N}\pi_i f_i
      - \frac{1}{2\delta}\sum_{i=1}^{N}
        \frac{\pi_i^2}{v_i}
    \right\},
  \qquad \delta > 0.
\end{equation}
The error between $\widetilde{\mathrm{AVaR}}_{\alpha,\delta}$ and $\mathrm{AVaR}_\alpha$ is $\mathcal{O}(\delta^{-1})$ and so $\widetilde{\mathrm{AVaR}}_{\alpha,\delta}\to\mathrm{AVaR}_\alpha$ as $\delta\to\infty$ \cite{kouri2020epi,kouri2022primal}.

\paragraph{Computation via weighted projection.}
The maximizer of \eqnref{eq:smooth-avar-E4} is the projection
of $s_i = \delta v_i f_i$ onto $\mathcal{C}_\alpha$ under the
$V^{-1}$-weighted norm, solved by weighted clipping:
\begin{equation}
  \pi_i^* = \min\!\left(\frac{v_i}{1-\alpha},\;
                   \max(0,\, s_i - \mu v_i)\right),
\end{equation}
where the Lagrange multiplier $\mu$ enforcing $\sum_i\pi_i^*=1$
is found by sorting the kinks of the piecewise linear function $\mu\mapsto\sum_i \pi_i^*-1$ and then using bisection to identify the interval containing the root.
The smoothed objective evaluates to
\begin{equation}
  \widetilde{\mathrm{AVaR}}_{\alpha,\delta}(f)
  = \sum_{i=1}^N \pi^*_i f_i
    - \frac{1}{2\delta}\sum_{i=1}^N \frac{(\pi^*_i)^2}{v_i}.
\end{equation}

\paragraph{Gradient.}
When the weights $v_i$ do not depend on $\design$, the gradient is
\begin{equation}
  \label{eq:avar-grad-fixed}
  \nabla_{\design}\,\widetilde{\mathrm{AVaR}}_{\alpha,\delta}(f)
  = \sum_{i=1}^N \pi^*_i \,\nabla_{\design} f_i.
\end{equation}
When $v_i = v_i(\design)$ (design-dependent importance weights),
no closed-form gradient is available owing to the discontinuous
active set of the binary search; gradients are propagated via
reverse-mode automatic differentiation through the smoothed
projection.
We use $\delta = 10^{5}$ throughout; the risk level $\alpha$ is
specified per experiment in \secref{sec:experiments}.

\subsection{Specialization: Std/Mean/Mean (design objective U1)}
\label{app:gradients:u1}

For $\devmeas = \mathrm{Std}$, $\qoirisk = \mathrm{Mean}$,
$\datarisk = \mathrm{Mean}$, and scalar QoI ($\nqoi = 1$),
the sample-average objective is
\begin{equation}
  \label{eq:u1-saa}
  \hat{U}_1(\design)
  = \sum_{m=1}^{\nout} \outwt{m}\,\sigma_m(\design),
\end{equation}
where $\outwt{m}$ are outer quadrature weights and
\begin{equation}
  \label{eq:std-dev}
  \sigma_m(\design) = \sqrt{V_m(\design)},
  \qquad
  V_m = M_m^{(2)} - \bigl[M_m^{(1)}\bigr]^2,
\end{equation}
with posterior moments
\begin{equation}
  \label{eq:moments}
  M_m^{(p)}(\design)
  = \sum_{n=1}^{\nin} \impwt{m}{n}(\design)\,f_n^p,
  \qquad p = 1, 2,
\end{equation}
and $f_n = \qoimap(\innersamp{n})$ the (fixed) QoI value at inner
sample $n$.

\paragraph{Moment gradients.}
Differentiating \eqnref{eq:moments} and substituting
\eqnref{eq:impweights-jac}:
\begin{equation}
  \label{eq:mom-jac}
  \frac{\partial M_m^{(p)}}{\partial\designk{k}}
  = \sum_{n=1}^{\nin} f_n^p\,
    \frac{\partial\impwt{m}{n}}{\partial\designk{k}},
  \qquad p = 1, 2.
\end{equation}

\paragraph{Variance and standard-deviation gradients.}
By the chain rule,
\begin{align}
  \frac{\partial V_m}{\partial\designk{k}}
  &= \frac{\partial M_m^{(2)}}{\partial\designk{k}}
     - 2\,M_m^{(1)}\,
       \frac{\partial M_m^{(1)}}{\partial\designk{k}},
  \label{eq:var-jac}\\[4pt]
  \frac{\partial\sigma_m}{\partial\designk{k}}
  &= \frac{1}{2\,\sigma_m}\,
     \frac{\partial V_m}{\partial\designk{k}}.
  \label{eq:std-jac}
\end{align}

\paragraph{Objective gradient.}
Combining \eqnref{eq:u1-saa} with \eqnref{eq:std-jac}:
\begin{equation}
  \label{eq:u1-grad}
  \frac{\partial\hat{U}_1}{\partial\designk{k}}
  = \sum_{m=1}^{\nout} \outwt{m}\,
    \frac{\partial\sigma_m}{\partial\designk{k}}.
\end{equation}
The gradients for other deviation measures (variance, entropic)
follow the same structure with $\sigma_m$ replaced by the
corresponding functional of the importance-weighted moments.


\section{Advection--Diffusion Experiment: Regime Replications}
\label{app:advdiff:regimes}

This appendix collects the diffusion-dominated ($D = 0.1$,
$\mathrm{Pe}_h = 0.024$) and moderate ($D = 0.02$,
$\mathrm{Pe}_h = 0.120$) regime replications of the
advection--diffusion sensor-placement experiment
(\secref{sec:experiments:advdiff:setup}).
The regimes differ only in the tracer diffusivity $D$ and are
summarized by the mesh P\'{e}clet number
$\mathrm{Pe}_h = |\vv{u}|_{\max}\, h / (2D)$.
Here $h$ is the local mesh element size and $|\vv{u}|_{\max}$ is
the maximum velocity magnitude from the Navier--Stokes solution.
The advection-dominated regime reported in the main text has
$D = 0.005$ and $\mathrm{Pe}_h = 0.484$.
Each regime is documented with a bootstrap-weight boxplot for the
active candidate sensors (\figref{fig:advdiff:boxplot:C1} and
\figref{fig:advdiff:boxplot:C2}) and a random-design histogram
with overlaid cross-evaluations of the six optimized designs
(\figref{fig:advdiff:histogram:C1} and
\figref{fig:advdiff:histogram:C2}).

\paragraph{Active sensors and bootstrap variability.}
The diffusion-dominated and moderate boxplots
(\figref{fig:advdiff:boxplot:C1}, \figref{fig:advdiff:boxplot:C2})
reproduce the qualitative structure of the advection-dominated
boxplot (\figref{fig:advdiff:boxplot}): a small subset of the
$\nobs = 50$ candidate sensors receives essentially all weight
under every design objective, the bootstrap boxes are narrow relative to
inter-objective gaps, and the design objectives disagree on both primary
and secondary sensor selection.
Two quantitative differences follow from the regime change.
First, the primary goal-oriented sensor shifts upstream as
transport becomes more diffusive: sensor~26, which anchors the
advection-dominated designs just past the first obstacle, loses
its dominance in favor of sensors placed farther left in the
domain, consistent with the contraction of the advective pathway.
Second, the inter-objective spread at each active sensor contracts.
In the advection-dominated regime the spread across the six
goal-oriented design objectives at the single dominant sensor approaches
$0.4$; in the moderate regime it is closer to $0.2$; and in the
diffusion-dominated regime the six goal-oriented design objectives
concentrate weight at the same two to three sensors with spreads
below $0.1$.

\paragraph{KL versus goal-oriented design separation.}
In the advection-dominated regime, the KL baseline places its
weight on sensors that every goal-oriented design objective ignores
(\secref{sec:experiments:advdiff:designs}); the two groups of dominant
sensors are disjoint (\figref{fig:advdiff:boxplot}).
As transport becomes more diffusive this separation narrows.
In the moderate regime the KL and goal-oriented top sensors
differ but overlap partially (\figref{fig:advdiff:boxplot:C2}); in
the diffusion-dominated regime the KL dominant sensors coincide
with the goal-oriented dominant sensors on the primary selection,
and the disagreement is confined to secondary allocations
(\figref{fig:advdiff:boxplot:C1}).
The physical mechanism is the same one responsible for the
sensor-shift discussed above: as diffusion broadens the
informative region, a single sensor can simultaneously constrain
the source parameters (which KL rewards) and the downstream
prediction (which the goal-oriented design objectives reward), so the two
notions of information become less orthogonal.

\paragraph{Practical cost of mismatched designs.}
The random-design histograms
(\figref{fig:advdiff:histogram:C1},
\figref{fig:advdiff:histogram:C2}) show the corresponding
shrinkage of the cross-evaluation penalty.
In the advection-dominated regime
(\figref{fig:advdiff:histogram}), the KL-optimal design
evaluated under any goal-oriented design objective sits deep inside the
random-design distribution --- comparable to an arbitrary
allocation on the simplex --- and the reverse holds for the six
goal-oriented designs under the KL design objective.
In the moderate regime the KL-optimal design moves toward the
favorable edge of the random-design histogram of each
goal-oriented design objective; its penalty relative to the matched
optimum is visible but far smaller than in the advection regime.
In the diffusion-dominated regime this penalty further
shrinks: the KL-optimal design and the six goal-oriented optima
all sit within approximately one histogram-width of the matched
optimum under any design objective, and the cross-evaluation asymmetry
between parameter-focused and prediction-focused designs
weakens.
The categorical error that misspecifying the objective class
incurs in the advection-dominated regime therefore softens to a
bounded quantitative error as diffusion broadens the information
available at each candidate.

\paragraph{Rounding to a deployable integer budget.}
A deployable design requires rounding the continuous optimum
$\design^\star \in \simplex$ to an integer measurement allocation
$\vv{n}^\star \in \mathbb{Z}_{\geq 0}^{\nobs}$ with
$\sum_k n^\star_k = S$, where $S$ is the total measurement
budget.
We propose using Hamilton's largest-remainder apportionment rule:
assign each sensor $\lfloor S w_k^\star \rfloor$ measurements,
then distribute the $S - \sum_k \lfloor S w_k^\star \rfloor$
remaining measurements to the sensors with the largest
fractional parts $S w_k^\star - \lfloor S w_k^\star \rfloor$.
\figref{fig:advdiff:budget} shows the result for the most
risk-averse design objective $U_3$ at three budgets
$S \in \{5, 10, 20\}$ in the advection-dominated regime.
The smallest budget places all five measurements on the
top-left pair of sensors, including multiple measurements at
sensor~26 where $U_3$ concentrates $\sim 0.7$ of its continuous
weight.
As $S$ grows the count at the dominant sensors rises first, and
only once those sites carry a commensurate number of
measurements do additional sensors activate to populate the
secondary support.
The rounded allocations therefore scale from highly concentrated
at small budgets to spatially diverse at larger budgets,
respecting the weight hierarchy of the continuous optimum without
a separate integer optimization.
This behavior is qualitatively unchanged across the three
regimes; only the underlying continuous optima, and hence the
specific sites populated, differ.
\begin{figure}[t]
  \centering
  \includegraphics[width=\textwidth]{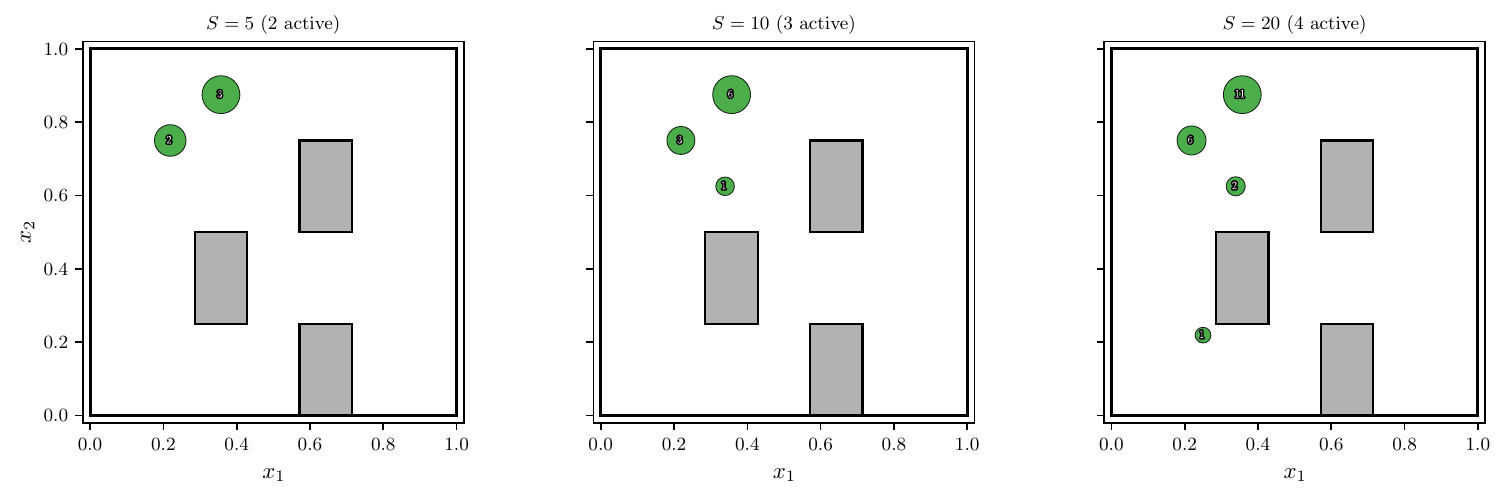}
  \caption{%
    Rounded integer measurement allocations for the $U_3$
    optimum in the advection-dominated regime at three budgets
    $S \in \{5, 10, 20\}$.
    Filled markers: active sensor locations, sized by integer
    measurement count at that location.
    }
  \label{fig:advdiff:budget}
\end{figure}
Increasing $S$ first grows the measurement count at
dominant sensors and then activates additional sensors
populating the secondary support of the continuous
optimum.%

\begin{figure}[t]
  \centering
  \includegraphics[width=\textwidth]{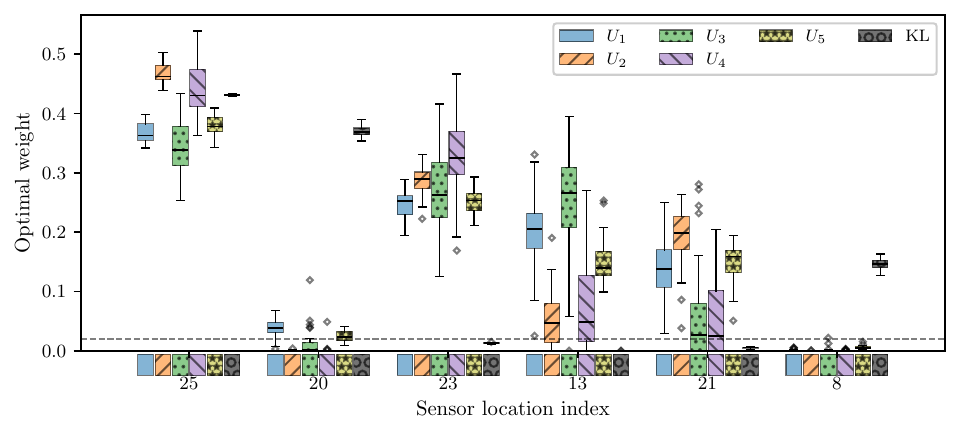}
  \caption{%
    Bootstrap distribution of optimal sensor weights in the
    diffusion-dominated regime ($D = 0.1$,
    $\mathrm{Pe}_h = 0.024$), across $30$ resampled training
    sets.
    Format matches \figref{fig:advdiff:boxplot}.%
  }
  \label{fig:advdiff:boxplot:C1}
\end{figure}

\begin{figure}[t]
  \centering
  \includegraphics[width=\textwidth]{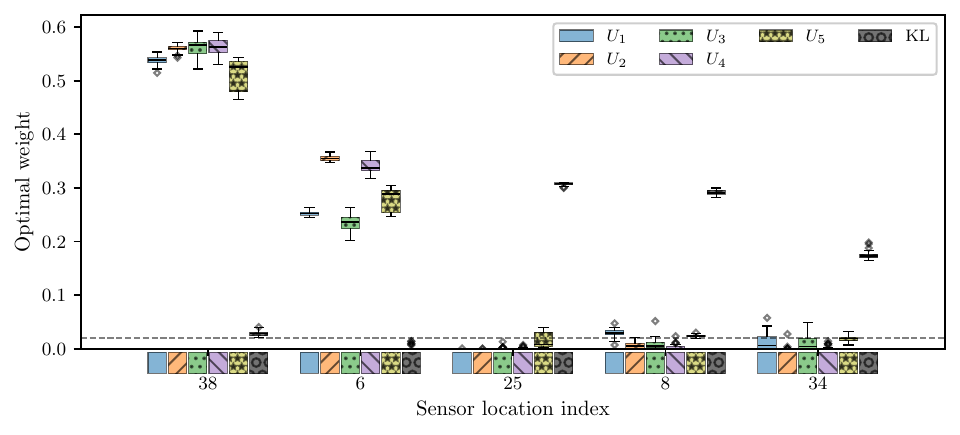}
  \caption{%
    Bootstrap distribution of optimal sensor weights in the
    moderate regime ($D = 0.02$, $\mathrm{Pe}_h = 0.120$),
    across $30$ resampled training sets.
    Format matches \figref{fig:advdiff:boxplot}.%
  }
  \label{fig:advdiff:boxplot:C2}
\end{figure}

\begin{figure}[t]
  \centering
  \includegraphics[width=\textwidth,height=3.0in,keepaspectratio]{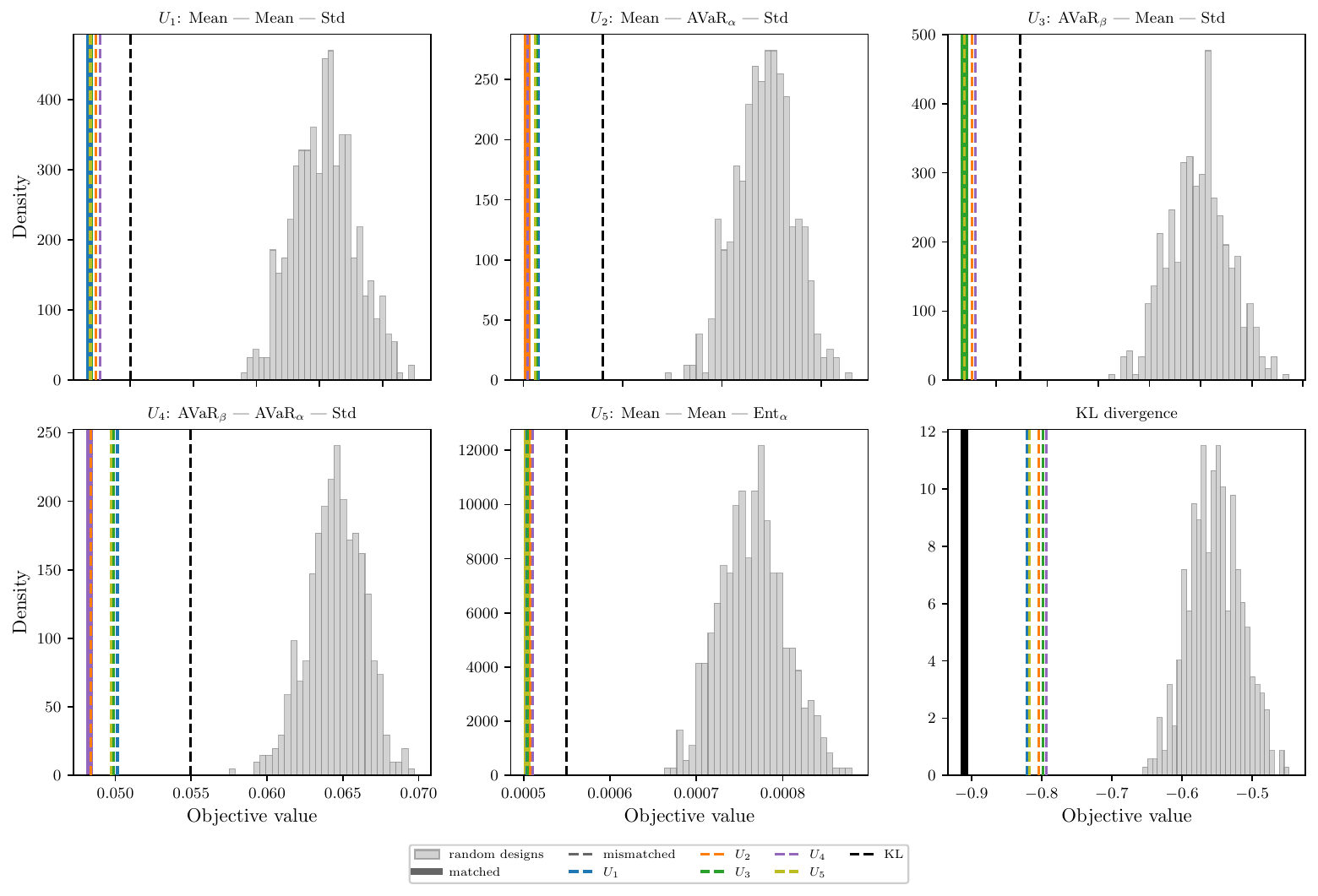}
  \caption{%
    Random-design histogram and optima cross-evaluation in the
    diffusion-dominated regime ($D = 0.1$,
    $\mathrm{Pe}_h = 0.024$).
    Format matches \figref{fig:advdiff:histogram}.%
  }
  \label{fig:advdiff:histogram:C1}
\end{figure}

\begin{figure}[t]
  \centering
  \includegraphics[width=\textwidth,height=3.0in,keepaspectratio]{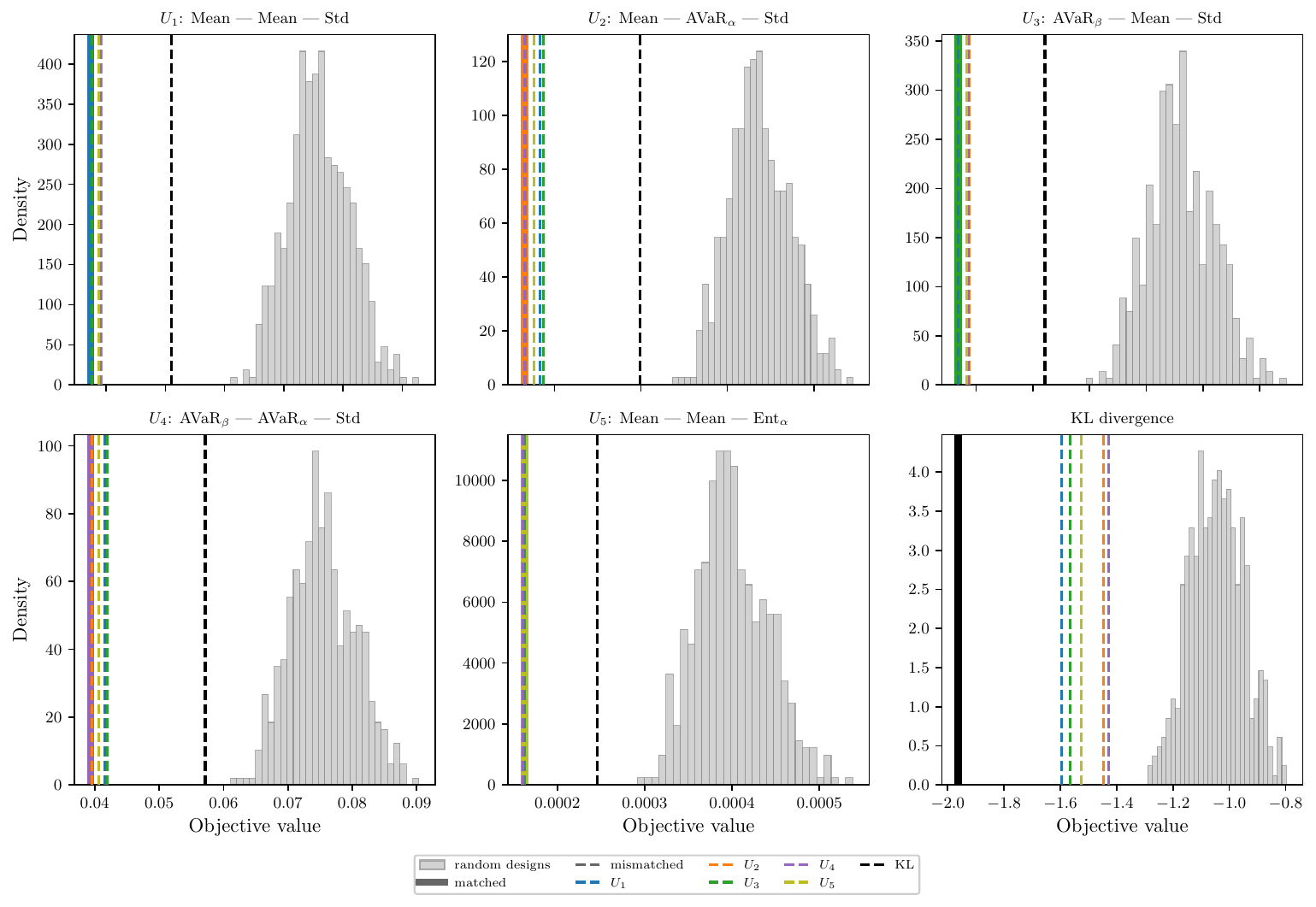}
  \caption{%
    Random-design histogram and optima cross-evaluation in the
    moderate regime ($D = 0.02$, $\mathrm{Pe}_h = 0.120$).
    Format matches \figref{fig:advdiff:histogram}.%
  }
  \label{fig:advdiff:histogram:C2}
\end{figure}

\section{Continuity of the Smoothed AVaR Value Function}\label{sec:app_smoothed_avar}

In this appendix, we verify the claim that the optimal smoothed AVaR path is continuous
with respect to the confidence parameter $\alpha\in[0,1)$. Suppose
$f:\Delta^K\to\mathcal{L}_1$, $r>0$, and $t\in[1,\infty)$, and consider the optimal value function
\[
  h(t) := \min_{w\in\Delta^K}\sup_{\theta\in\mathfrak{A}_t} \left\{\mathbb{E}[\theta f(w)] - \tfrac{1}{2r}\mathbb{E}[\theta^2]\right\},
\]
where
\[
  \mathfrak{A}_t := \{\theta\in\mathcal{L}_\infty\,\vert\,\mathbb{E}[\theta]=1,\;\; 0\le\theta\le t\;\;\text{a.s.}\}
\]
and $\mathcal{L}_1$ denotes the Lebesgue space of absolutely integrable random variables 
and $\mathcal{L}_\infty$ denotes the Lebesgue space of essentially bounded
random variables.
Recall that the inner maximization problem is the dual representation of the
smoothed AVaR.  Ultimately, we will show that $\alpha\mapsto h(1/(1-\alpha))$ is
continuous. To achieve this, we will first show
that $h(\cdot)$ is finite-valued, concave and increasing, implying that it is continuous.  The
target function is then continuous since it is the composition of continuous
maps.  To begin, we note that for any $t_1,t_2\in[1,\infty)$ and $\lambda\in[0,1]$,
we have that
\[
  \lambda\mathfrak{A}_{t_1}+(1-\lambda)\mathfrak{A}_{t_2}\subseteq\mathfrak{A}_{\lambda t_1+(1-\lambda)t_2}.
\]
Here, the sum on the left-hand side is the usual Minkowski sum of sets.  To
prove this inclusion, let $\theta_1\in\mathfrak{A}_{t_1}$ and
$\theta_2\in\mathfrak{A}_{t_2}$, then
\[
  \mathbb{E}[\lambda\theta_1+(1-\lambda)\theta_2]=\lambda\mathbb{E}[\theta_1]+(1-\lambda)\mathbb{E}[\theta_2]=\lambda+(1-\lambda)=1
\]
and
\[
  0\le\lambda\theta_1+(1-\lambda)\theta_2\le\lambda t_1+(1-\lambda)t_2 \quad\text{a.s.}
\]
Hence, $\lambda\theta_1+(1-\lambda)\theta_2\in\mathfrak{A}_{\lambda t_1+(1-\lambda)t_2}$.
Using this inclusion, we then have that
\[
\begin{aligned}
  h(\lambda t_1+(1-\lambda)t_2) &\ge \min_{w\in\Delta^K} \left\{\mathbb{E}[(\lambda\theta_1+(1-\lambda)\theta_2) f(w)] - \tfrac{1}{2r}\mathbb{E}[(\lambda\theta_1+(1-\lambda)\theta_2)^2]\right\} \\
  &\ge \min_{w\in\Delta^K}\{\lambda(\mathbb{E}[\theta_1 f(w)]-\tfrac{1}{2r}\mathbb{E}[\theta_1^2])+(1-\lambda)(\mathbb{E}[\theta_2 f(w)]-\tfrac{1}{2r}\mathbb{E}[\theta_2^2])\}
\end{aligned}
\]
for all $\theta_1\in\mathfrak{A}_{t_1}$, $\theta_2\in\mathfrak{A}_{t_2}$ and $\lambda\in[0,1]$, where the second inequality follows from the concavity of $\theta\mapsto-\tfrac{1}{2r}\mathbb{E}[\theta^2]$.
Maximizing the right-hand side over $\theta_1$ and $\theta_2$, and noting that
the minimum of a sum is greater than or equal to the sum of the individual
minimums demonstrates that $h(\cdot)$ is concave.  
Additionally, since $\mathfrak{A}_{t_1}\subseteq\mathfrak{A}_{t_2}$ for all $t_2\ge t_1\ge 1$, we have that the supremum over $\mathfrak{A}_{t_2}$ is greater than or equal to the supremum over $\mathfrak{A}_{t_1}$ and therefore $h(\cdot)$ is increasing.
Finally, since $f(w)\in\mathcal{L}_1$ for all $w\in\Delta^K$, $h(\cdot)$ is finite and these three properties ensure that $h(\cdot)$ is continuous on $[1,\infty)$. Consequently, $\alpha\mapsto h(1/(1-\alpha))$ is continuous and increasing on $[0,1)$ since it is the composition of continuous and increasing maps.
\else

\fi

\end{document}